\documentclass[twocolumn,pra,aps]{revtex4}
\usepackage{enumitem}
\usepackage{graphicx}
\usepackage{bm}
\usepackage{amsmath}
\usepackage{amssymb}
\usepackage{underscore}
\usepackage{color}
\usepackage{float}
\usepackage[caption = false]{subfig}
\usepackage{amsmath,amsthm,amsfonts,amssymb,amscd}

\def\squareforqed{\hbox{\rlap{$\sqcap$}$\sqcup$}}
\def\qed{\ifmmode\squareforqed\else{\unskip\nobreak\hfil
\penalty50\hskip1em\null\nobreak\hfil\squareforqed
\parfillskip=0pt\finalhyphendemerits=0\endgraf}\fi}

\def\duzomniejsze{<\kern-.7mm<}
\def\duzowieksze{>\kern-.7mm>}

\def\textbf#1{{\bf #1}}
\def\beq{\begin{equation}}
\def\eeq{\end{equation}}
\def\be{\begin{equation}}
\def\ee{\end{equation}}
\def\ben{\begin{eqnarray}}
\def\een{\end{eqnarray}}
\def\beqa{\begin{eqnarray}}
\def\eeqa{\end{eqnarray}}
\def\eea{\end{array}}
\def\bea{\begin{array}}

\newcommand{\bei}{\begin{itemize}}
\newcommand{\eei}{\end{itemize}}
\newcommand{\bee}{\begin{enumerate}}
\newcommand{\eee}{\end{enumerate}}
\newcommand{\nc}{\newcommand}
\nc{\1}{{\openone}}

\def\>{\rangle}
\def\<{\langle}

\newtheorem{lemma}{Lemma}
\newtheorem{cor}{Corollary}

\newtheorem{theorem}{Theorem}
\newtheorem{definition}{Definition}
\newtheorem{obs}{Observation}

\newtheorem{fact}{Fact}

\def\bed{\begin{definition}}
\def\eed{\end{definition}}
\def\bel{\begin{lemma}}
\def\eel{\end{lemma}}
\def\bet{\begin{theorem}}
\def\eet{\end{theorem}}
\def\be{\begin{equation}}
\def\ee{\end{equation}}

\begin{document}
\title{Linear game non-contextuality and Bell inequalities - a graph-theoretic approach}
%\title{Color Bell and contextual inequalities - generalized XOR games}
\author{P. Gnaci\'nski$^{1}$, M. Rosicka$^{1,2}$,  R. Ramanathan$^{1}$, K. Horodecki$^{3}$, M. Horodecki$^{1}$, P. Horodecki$^{2}$, S. Severini$^4$}
\affiliation{$^1$
Faculty of Mathematics, Physics and Informatics, University of Gda\'{n}sk, 80-952 Gda\'{n}sk,  Institute of Theoretical Physics and Astrophysics, and
National Quantum Information Centre in Gda\'{n}sk, 81-824 Sopot, Poland}
\affiliation{$^2$Faculty of Applied Physics and Mathematics, Gda\'{n}sk University of Technology, 80-233 Gda\'{n}sk, Poland
 and National Quantum Information Centre in Gda\'{n}sk, 81-824 Sopot, Poland}
\affiliation{$^3$Faculty of Mathematics, Physics and Informatics, University of Gda\'{n}sk, 80-952 Gda\'{n}sk, Institute of Informatics, and National Quantum Information Centre in Gda\'{n}sk, 81-824 Sopot, Poland}
\affiliation{$^4$Department of Computer Science and Department of Physics and Astronomy,
University College London, WC1E 6BT London, U.K}

\begin{abstract}

We study the classical and quantum values of one- and two-party linear games, an important class of unique games that generalizes the well-known XOR games to the case of non-binary outcomes. We introduce a ``constraint graph" associated to such a game, with the constraints defining the linear game represented by an edge-coloring of the graph. We use the graph-theoretic characterization to relate the task of finding equivalent games to the notion of signed graphs and switching equivalence from graph theory. We relate the problem of computing the classical value of single-party anti-correlation XOR games to finding the edge bipartization number of a graph, which is known to be MaxSNP hard, and connect the computation of the classical value of more general XOR-d games to the identification of specific cycles in the graph. We construct an orthogonality graph of the game from the constraint graph and study its Lov\'{a}sz theta number as a general upper bound on the quantum value even in the case of single-party contextual XOR-d games. Linear games possess appealing properties for use in device-independent applications such as randomness of the local correlated outcomes in the optimal quantum strategy. We study the possibility of obtaining quantum algebraic violation of these games, and show that no finite linear game possesses the property of pseudo-telepathy leaving the frequently used chained Bell inequalities as the natural candidates for such applications. We also show this lack of pseudo-telepathy for multi-party XOR-type inequalities involving two-body correlation functions.

\end{abstract}

\maketitle
\section{Introduction}
\bibliographystyle{apsrev}

Quantum mechanics provides various resources. One of them is quantum non-locality \cite{Bell, review-nonloc}. Given the ability to perform measurements on a bipartite quantum state, one can obtain  correlations which do not have a classical explanation in that they can not be predetermined before the measurements. To ensure this, one can perform statistical tests for quantum non-locality \cite{Bell}, known as the Bell inequalities, the famous CHSH \cite{CHSH} inequality being a prominent example. The applications of non-locality go beyond quantum theory \cite{BHK, E91}, reaching as far as device-independent security against a so called non-signaling adversary - a person possibly empowered with more than quantum resources, but still obeying the no-faster-than-light communication principle \cite{PR}. Another application of quantum non-locality is to communication complexity \cite{BCMdW10}, where the use of quantum non-local correlations lowers the communication cost of evaluating a function using distributed computers.

Bell non-locality is a special case of the general phenomenon called {\it contextuality}. This phenomenon which had been discovered first by Kochen and Specker \cite{Kochen-Specker} stems from the fact that in quantum mechanics the results of the measurement of an observable may depend on the context (i.e., the particular set of commuting observables) in which it is measured.
% that is, on the set of other observables (e.g. either $\{A,B,C\}$ or $\{A,D,F\}$) which commutes with $A$ (here $B,C$ commutes with $A$ and $D,F$ so, but e.g. $D$ does not commute with $C$). 
In consequence even for a single quantum system, sometimes a measurement can be said to create the outcomes, instead of merely revealing preexisting ones. 
Quite a long history of research on contextuality has led to various non-contextuality inequalities \cite{Kochen-Specker, KCBS, Cabello, CSW, AQBCC13}, Bell inequalities being a special case. Quantum contextuality has for long been studied as a fundamental quantum property, reaching recently a connection to a resource which is required for universal quantum computing \cite{Contex-Nature,Delf-context-rebits} and quantum cryptography \cite{context-measures}.

Two-party Bell inequalities have also been studied in theoretical computer science in terms of two-prover interactive-proof systems, commonly referred to as "`games" \cite{Condon} between two players and a referee. In this formulation one can let the players pre-share quantum data (an entangled quantum state) and the use of outcomes of measurements on it can lead to a higher probability of winning the game than in the case of classical shared randomness. Yet higher success probability may be obtained, when the players are provided with a general system (device) which is only required to satisfy the no-signaling principle. In this framework, the main quantity of interest is the winning probability of the game or in general the amount of violation of a  Bell inequality. In the case of a single player, the Bell inequality becomes a non-contextuality inequality or simply
%such a case has also been studied in computer science and is known under the name of 
a constraint satisfaction problem (see e.g. \cite{Trevisan} and references therein). 

In general it is NP-hard to find the classical value of a general constraint satisfaction problem with many variables per constraint 
\cite{AroraLMSS98,Arora-Safra,Raz98}, so one considers special classes of games. A celebrated class of games is the so-called {\it unique games} with two players. These are games where for each pair of questions by the referee $(x,y)$ and for any answer of one player $a$ there exists only one answer of the other player $b$ which leads to winning. In other words, the winning constraints are permutations: one-to-one mappings of the answers of one player into acceptable answers of the other: $\pi_{(x,y)}(a)=b$. 
%For this reason, research has been focused on simple classes of the games. 
%A celebrated one is the class called {\it unique games}. These are games where for each pair of questions by the referee $(s,t)$ and for any answer of one player $a$ there exists only one answer of the other player $b$ which leads to winning (i.e., satisfies constraint $q$). In other words, the constraints are permutations: mapping the answers of one player one-to-one into acceptable answers of the other: $\pi_{s,t}(a)=b$.
Computing the exact classical value of a unique game is known to be NP hard \cite{Hastad}. Moreover, it is conjectured, that it is even NP hard to distinguish whether a unique game has classical probability of winning almost $1$, or close to zero. This conjecture, known as the Unique Games Conjecture, has vast consequences for many questions in computer science \cite{Khot02}.
On the other hand, it is known that the quantum winning probability of the unique game can be approximated to within a constant factor in polynomial time \cite{KempeRegevToner}, in particular for a unique game with quantum value $1 - \epsilon$, one can find in polynomial time (in the number of inputs and outputs of the game) an entangled strategy which achieves value at least $1 - 6 \epsilon$ for the game. A subclass of unique games are the so-called XOR games for two players, where the players return binary answers and the winning constraint for the game only depends on the XOR of the players' answers.
%have binary answers and whether a pair of answers is accepted by the referee is only a function of their XOR. 
Computing the classical value of even this simplest class of unique game turns out be NP hard \cite{Hastad}, however it is known from the results of \cite{Tsirelson, Cleve} that the exact quantum value of the two-party XOR game can be computed in polynomial time. It is notable that the XOR games are equivalent to correlation based Bell inequalities for two outcomes and have also been extensively studied in the physics literature \cite{CHSH, BraunsteinC1988, Mermin}. As such, virtually all applications of quantum non-locality such as in device-independent cryptography \cite{BHK, E91} or randomness generation \cite{MS} use two-player XOR games or their multi-party generalization in terms of GHZ paradoxes \cite{Mermin}.  
% note that the classical value of XOR games are still hard to compute \cite{Hastad}. 

%The most well studied and known game is the XOR game \cite{CleveHoyTonWat}, where the winning of the game depends solely on the value of XOR of the outputs of measurements of the players. In this case, it is known, that the highest probability of successs assisted by quantum data can be obtained efficiently by SDP program \cite{KempeRegevToner,RegevVidick}. On the other hand finding the one in case of classical data assistance, is an NP-hard problem .  

While XOR games have found widespread use, recently there has been much interest in developing applications of higher-dimensional entanglement \cite{Exp-high-dim, Qudit-Toffoli, Qudit-key-dist} for which Bell inequalities with more than two outcomes are naturally suited. Therefore, both for fundamental reasons as well as for these applications, the study of Bell and non-contextuality inequalities with more outcomes is crucial. In this paper, we study a natural generalization of XOR games which we call 
%Motivated by both possible applications as well as the success in studying these two fundamental resources we study in this paper the %Bell and non-contextuality inequalities, focusing on games which we call 
generalized XOR (GXOR) games or XOR-d games \cite{RAM}. Such games in case of two ternary inputs per party appeared
first in the context of experiments \cite{ZZH}, the specific example of the generalized CHSH game was studied in \cite{Buhrman, BavarianShor, RAM} and a general bound on the quantum value of XOR-d games was proposed in \cite{RAM}. 
% however has not been studied in depth so far. 
%Special cases of these games have been studied in \cite{BavarianShor, RAM}, where analytical bounds on wining probability on quantum resources has been found.
%We study the {\it exact} maximal values for the classical, quantum and non-signaling case, and our approach is to provide a graph-theoretical characterization of the problem. 
In this paper, we introduce a graph-theoretic characterization of these games, and apply it to the problem of finding the maximal classical and quantum values of such games.

The paper is organized as follows. The section \ref{sec:XOR-and-beyond} introduces the graph-theoretic formulation of XOR games and the expression of the game value using graph-theoretic invariants involving edge labeling. 
%in both traditional terms of non-signaling conditional probability distribution as well as its graph-theoretical approach. 
%In particular in section \ref{subsec:XORgames} we explain some properties of XOR game. 
We then describe an axiomatic generalization of the XOR games in terms of two properties and show that the previously defined class of linear games \cite{Hastad} is the unique class which satisfies these properties. We subsequently establish the graph-theoretic characterization of the subset of XOR-d games and illustrate this with the example of games with ternary outputs.
% In particular, the questions in the game correspond to the vertices of a graph and the winning constraint is expressed in terms of an edge labeling of the graph (using colored edges).  
% GXOR game, and show it graphical representation. We exemplify this in case of games with three outputs which (in case of wining the game) should follow one of $3$ types of correlations which are depicted as {\it colors}: red, green or blue. Such a game is depicted in form of graph, whose vertices correspond to questions and edges have one of these three colors, denoting appropriate correlations that rules which answers are correct.
%In \ref{subsec:singleGENXORgame}, we describe in the same manner the contextual (single party) games and discuss their classical, quantum and superaquantum values. 
%In \ref{subsec:GOXR-uncolored} we show how to incorporate partial function linear games within the graph formalism using \textit{uncolored edges}, the role of which is to simply identify that the two corresponding measurements commute while no winning constraint is imposed on their outputs. 
%we introduce the so called {\it uncolored} edges, the role of which is to identify that two measurements commute, but otherwise no constraint for the correlations of their outputs is demanded. 
We then describe one of our results in section \ref{sec:labeled-eq}, where we use the graph-theoretic formalism established in previous sections to identify when two games can be considered equivalent, in particular we establish a relation to the graph-theoretic notion of signed graphs and switching equivalence. 
% describes under what conditions two graphs will be considered as {\it equivalent}, that is describing games with the same classical, quantum and super-quantum values. 
Then in section \ref{sec:bin-xor} we study the classical value of these generalized XOR-d games in a graph-theoretical manner. Our results in this section include a characterization of the complexity (as MaxSNP-hard) of computing the classical value of the simplest class of XOR games, namely single-party anti-correlation games.  
% studying among others the cycles.
In the next section \ref{sec:Lovas}, we study the quantum value of these games, in particular we establish that the well-known Lov\'{a}sz theta number of the orthogonality graph of a contextuality game only gives an upper bound to its quantum value, unlike in the previously considered scenario of non-contextuality inequalities involving rank-one projectors. XOR-d games have the important property that their optimal quantum strategies involve locally random and correlated outcomes, thus permitting them to be ideal candidates for device-independent applications. In section \ref{sec:DIapp}, we prove that no non-trivial finite XOR-d game for prime $d$ can be perfectly won with a quantum strategy, thus providing evidence that the frequently used chained Bell inequalities might indeed be the best candidates for such applications. We also extend the result to multi-party "partial" XOR games which involve only two-body correlation functions, showing that such Bell inequalities cannot achieve algebraic violation. 
%connect it with the Lov\'{a}sz theta value of the corresponding orthogonality graph.
%the graph theoretical approach to quantum value of games via the Lov\'{a}sz theta number is presented. 
The final section \ref{sec:Comparing-cl-qn} is devoted to a numerical analysis of the classical and quantum values (using semi-definite programming) of games with upto three ternary inputs per party. We end with conclusions and some open problems.
%
%numerical results of investigation of classical and quantum values via the hierarchy of semidefinite programs approach in case of both single color and many color games. Finally we present conclusions and 
%some open problems relating the color Bell and contextual inequalities.
\section{Graph-theoretic formulation of generalized XOR games}
%From XOR to Generalized XOR games and beyond - introduction}
\label{sec:XOR-and-beyond}
The aim of this section is to introduce the Generalized XOR games in a graph theoretical manner. 
In order to do it, let us first recall a formulation of binary outcome XOR games in terms of graphs with two types of edges corresponding to correlated and anti-correlated answers in section \ref{subsec:XORgames}. Specifically, the constraints will be represented by two differently labeled edges on a graph with vertices representing the questions to the players so that the graph is a bipartite graph. We then define the main objects of study - the winning probabilities of a game given
classical, quantum and super-quantum resources respectively. 
%An XOR game can be depicted in terms of graph with two types of edges, corresponding to correlated and anti-correlated answers, which we recall in section \ref{subsec:XORgames}. Bearing this in mind, in section
In section \ref{subsec:GENXORgames}, we define the generalized XOR (XOR-d) games and establish their graph-theoretic formulation. The constraints of the game are represented by colored edges (with more than two colors), we illustrate this with the example of games with ternary answers. 
%. We exemplify this by ternary answers with different number of questions, where three types of edges (colors)
%are allowed. 
We then use the graph-theoretic formulation to also represent a single player contextuality game. This is simply a {\it constraint satisfaction problem}: the constraints of the game still being represented by colored edges, but with no bi-partition on the vertices. 
%that are put on
%the observables which are at hand of {\it single} person.

\subsection{XOR games}
\label{subsec:XORgames}

The XOR game involves a referee and two players: Alice and Bob. The referee asks questions $x \in \textsl{X}$ to Alice and $y \in \textsl{Y}$ to Bob according to an input probability distribution $\pi(x,y)$. Each player has two possible answers $a, b\in\{0,1\}$ respectively. Whether the players win or lose depends solely on the XOR of their outputs: $a\oplus b$, where $\oplus$ denotes addition modulo $2$. To give an example, in the famous CHSH game, the players win if $a\oplus b = x \cdot y$, i.e., when the $XOR$ of their answers equals the $AND$ of the questions, with $a,b,x,y$ being binary. XOR games are equivalent to correlation Bell inequalities with binary outcomes, since the correlation functions $\mathcal{E}_{x,y}$ are simply given by $\mathcal{E}_{x,y} = \sum_{k=0,1} (-1)^k P(a \oplus b = k | x, y)$.

In the game, the players can have access to certain resources. Three types of resources are usually considered. The first are classical corresponding to shared randomness between the players. The second are quantum, i.e., access to a bipartite entangled quantum state, and the set of measurements that can be performed on it by each player. And finally one also considers super-quantum resources, which correspond to access to a general device with inputs and outputs with the only constraint being that the device does not allow for signaling between the players. All the three resources have a common mathematical formulation as a conditional probability distribution $P(a,b|x,y)$ from a certain set: classical ($C$), quantum ($Q$) and super-quantum $(SQ)$, and in general $C \subset Q \subset SQ$. The no-signaling condition is expressed mathematically as
\ben
\sum_a P(a,b|x,y) &=& \sum_a P(a,b|x',y) \; \; \forall_{x,x',y,b}\, \nonumber\\
\sum_b P(a,b|x,y) &=& \sum_b P(a,b|x,y') \; \; \forall_{y,y',x,a}\,,
\label{eq:nscondition}
\een
in other words, the conditional probability distribution of each player is independent of the other party's input. 
The main object of study in XOR games is the winning probability of the players, which is written as:
%\begin{widetext}
\ben
\label{eq:xor-game}
\omega_S(G) =  \max_{P \in S} \sum_{\substack{x \in \textsl{X},\\ y \in \textsl{Y}} }\pi(x,y)\sum_{a,b \in \{0,1\}}V(a,b|x,y)P(a,b|x,y)
\een
%\end{widetext}
where $S \in \{ C, Q, SQ\}$ and $V(a,b|xy)$ is the indicator function reporting if the answers are correct (for XOR games $V(a,b|x,y)$ only depends on $a \oplus b$). For example, in the case of the CHSH game 
$V(a,b|x,y) = 1$ if $a\oplus b = x \cdot y$ and is set to $0$ otherwise. 
The three quantities are accordingly called the classical, quantum and super-quantum value of the game. 

%XOR games are equivalent to correlation Bell inequalities with binary outcomes, since the correlation functions $\mathcal{E}_{x,y}$ are simply given by $\mathcal{E}_{x,y} = \sum_{k=0,1} (-1)^k P(a \oplus b = k | x, y)$. 
%We are now ready to present the formulation of XOR games in graph-theoretic terms, using edge labelings of the graph to denote the winning constraint of the game. 
%in order to show explicitly how it differs from the case of generalized XOR games (GXOR). 
%For an XOR game depicted by a graph $G$, we will use $ p_{win}(S,G)$ to denote the winning probability of the game using the resource set $S$ under a uniform input distribution.

\subsubsection{Graph-theoretic formulation of XOR games}
We are now ready to present the formulation of XOR games in graph-theoretic terms.
%using edge labelings of the graph to denote the winning constraint of the game. 
%in order to show explicitly how it differs from the case of generalized XOR games (GXOR). 
An XOR game is represented by a 
%(commutation) 
graph $G$ with a specific edge-labeling that denotes the winning constraint of the game. The inputs, i.e., the questions asked by the referee, are represented by the vertices of a graph $G$. Two inputs are adjacent in the graph (i.e., connected by an edge) if and only if the corresponding measurements can be performed simultaneously. The winning constraint $V(a,b|x,y)$ in Eq. (\ref{eq:xor-game}) is represented by two types of edges - a solid edge corresponding to $a\oplus b=0$ (perfect correlations between the players) and a dashed edge corresponding to $a\oplus b=1$ (perfect anti-correlations between the players) that connect the inputs $x$ and $y$. A Bell inequality is thus represented by a bipartite graph (with the bi-partition corresponding to the two players). 

%Every XOR game is a {\it unique} game i.e., for every pair of questions $(x,y)$ and any answer of one player $a$ there is a unique answer $b$ of the second player that leads to winning. For this reason, we can depict the two kinds of correlations also as permutations of the set of outcomes. Correlations are then denoted by the identity $\mathbb{I}$ and anti-correlations by the transposition $(01)$ (see Fig \ref{fig:2permutacje}).
%, see for example the depiction of the CHSH game in Fig \ref{fig:CHSH}. 
%\begin{figure}[h!]
%\subfloat[fig 1]{\includegraphics[width = 2in]{2permutacje.pdf}} \\
%\subfloat[fig 2]{\includegraphics[width = 2in]{CHSH.pdf}}
%\caption{Add your own figures before compiling}
%\label{some example}
%\end{figure}

%\begin{figure}[h!]
%\begin{subfigure}{.5\textwidth}
%  \centering
%  \includegraphics[width=.4\linewidth]{2permutacje.pdf}
%  \caption{1a}
%  \label{fig:sfig1}
%\end{subfigure}%
%\begin{subfigure}{.5\textwidth}
%  \centering
%  \includegraphics[width=.2\linewidth]{CHSH.pdf}
%  \caption{1b}
%  \label{fig:sfig2}
%\end{subfigure}
%\caption{plots of....}
%\label{fig:fig}
%\end{figure}
%\begin{center}
%\begin{figure}[h!]
%		\includegraphics[width=0.2\textwidth]{CHSH.pdf}
%\caption{(Color online) Graphical representation of a CHSH game} 
%	\label{fig:CHSH}
%\end{figure}
%\end{center}

Every XOR game is a {\it unique} game i.e. for every pair of questions $(x,y)$ and an answer of one player $a$ there is a unique answer $b$ of the second player that leads to winning. For this reason, we can also depict the two kinds of correlations as permutations of the set of outcomes. Correlations are denoted by the identity $\mathbb{I}$ (i.e. $\pi(a) = a)$) and anti-correlations by the transposition $(01)$ (i.e. $\pi(a) = a \oplus 1 \; \text{mod 2}$) (see Fig \ref{fig:2permutacje}). We can formally define this as a labeling $K:E(G)\mapsto \{\mathbb{I}, (01)\}$ of the edges of graph $G$. For an XOR game depicted by a graph $G$ with an edge-labeling $K$, we use $ \omega_{S}(G, K)$ to denote the winning probability of the game using the resource set $S$ under a uniform input distribution.
% with permutations of the set of possible outcomes $\{0,1\}$. 
%The identity (i.e. $\pi(a)=a$) represents a correlation while the transposition $(01)$, (i.e. $\pi(a)=1-a\mod 2$) represents an anticorrelation.
%Note that $\mathbb{I}(a)+a=0 \mod 2$ and $(01)(a)+ a = 1\mod 2$ for $a\in\{0,1\}$.

\begin{center}
\begin{figure}[h!]
		\includegraphics[width=0.4\textwidth]{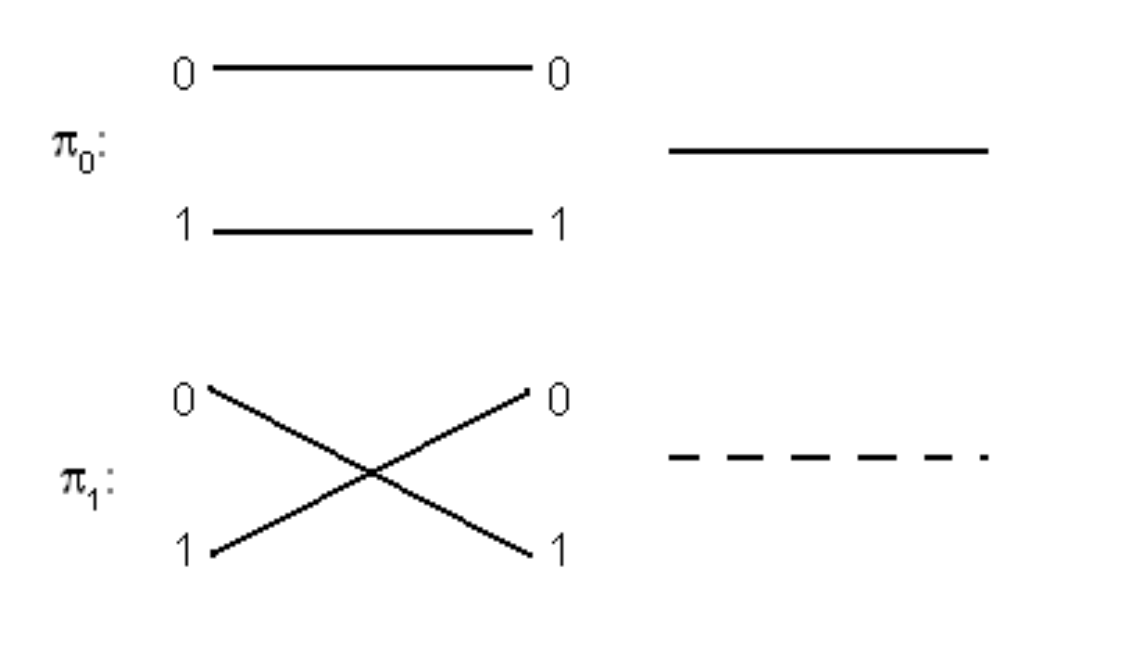}
\caption{(Color online) Two permutations involved in graphs of XOR games: $\pi_0=\mathbb{I}$ denotes correlations and $\pi_1=(01)$ denotes anti-correlation of the (binary) outputs.}
	\label{fig:2permutacje}
\end{figure}
\end{center}

\subsection{Generalized XOR (XOR-d) Games}
\label{subsec:GENXORgames}
In the generalization of an XOR games to games with $d$ outcomes, we abstract two properties of the XOR game: we require a set of $d$ permutations of $[d] :=\{0,1,...,d-1\}$ to describe the possible winning constraints in the game and impose that these permutations satisfy two salient properties:
%
%the set of $d$ permutations of $[d] :=\{0,1,...,d-1\}$ that describe the possible winning constraints in the game satisfy two properties: 

\begin{itemize}
%[label=(\subscript{P}{\arabic*})]
%\label{properties}
\item (P1) Each permutation is symmetric with respect to exchange of players, i.e. the permutations are their own inverse. 

\item (P2) Every pair $(a,\pi(a))$ appears exactly once in the set of permutations (in particular, each  permutation assigns a different $\pi(a)$ for each given $a\in [d]$.)
\end{itemize}

For instance, observe that the following set of permutations satisfies the above properties. For each answer $a$ of Alice, consider an answer of Bob as $b = \pi_i(a)$ where $\pi_i $  satisfies relation:
\begin{equation}
\label{eq:def-gxor}
\pi_i(a) + a = i \mod d
\end{equation}
for each $i \in [d]$ where $d$ is the number of possible answers for both players. Thus all permutations $\pi_i$ belong to the set 
%\com{I put below $x \in [d]$ is that ok ?}
\begin{equation}
\label{eq:def-gxor2}
L_d=\{\pi_i\in S_d: \pi_i(x)=i-x \; \text{mod d} \; \; \,\,\forall\,\, i,x \in [d]\}.
\end{equation}
where $S_d$ is the set of permutations of the set $[d]$. Note that these games belong to the class of linear games studied in \cite{Hastad, RAM} where the answers $a, b$ of the parties are required to obey $a + b \; \text{ mod d} = f(x,y)$ for a given set of functions $f(x,y)$ that characterize the game. 

Let us now prove that (up to local relabeling of answers) for odd $d$, up to local relabeling of answers the above set of permutations in Eq. (\ref{eq:def-gxor2}) is the only one which satisfies the two properties above, i.e., that the two properties (P1) and (P2) completely characterize the XOR-d game. For even $d$, this is no longer the case. Note that the prime $d$ corresponds to the case where the operation a + b mod d in Eq.(\ref{eq:def-gxor}) is the addition in a finite field $\mathbb{F}_d$, which is the interesting case of linear games studied for example in \cite{Hastad}. 

\begin{theorem}
For odd $d$, up to local relabeling of answers by the parties, the only games which satisfy the properties $P_1$ and $P_2$ are those given by Eq. (\ref{eq:def-gxor}).  
\end{theorem} 

\begin{proof}
The permutations that obey $P_1$ clearly have cycles of length at most two, i.e., they consist of fixed points and transpositions only. Let us first note that a permutation consisting of an even number of fixed points cannot be part of the set of permutations considered, because the permutations consists only of transpositions besides the fixed points. Also, a permutation consisting of an odd (greater than one) number of fixed points cannot be part of the set of permutations considered. This is because of the requirement that there be $d$ permutations in the set and each permutation consists of at least one fixed point due to the previous considerations, so that having a permutation with more than one fixed point in the set leads to a contradiction with $P_2$. We therefore see that each permutation in the set of $d$ permutations contains exactly one (distinct) fixed point. 

Now, the number of permutations $M_d$ of a set of $d$ objects (with $d$ odd) consisting of only transpositions and exactly one fixed point is given by
\begin{equation}
\label{eq:trans-fix}
M_d =  \frac{d}{\left(\frac{d-1}{2} \right)!} \prod_{j=0}^{(d-3)/2} \binom{d-2j-1}{2}
\end{equation}
Now, the $d$ permutation set $L_d$ defined by Eq.(\ref{eq:def-gxor2}) clearly obeys $P_1$ and $P_2$. Also, any set of $d$ permutations obtained from $L_d$ by a local relabeling, i.e., applying a permutation $\Pi \in S_d$ to each element $a$ and $\pi_i(a)$ for all $i = 0,...,d-1$ and $\pi_i \in L_d$ also obeys $P_1$ and $P_2$. The number of permutations obtained by this operation is $d! \times d$ but this involves over counting since many of the permutations thus obtained are equivalent. A precise counting argument shows that the exact number of permutations obtained is given by
\begin{equation}
\frac{d!}{2^{(d-1)/2} \times \left(\frac{d-1}{2} \right)!}
\end{equation}  
which after some algebra is seen to be exactly equivalent to $M_d$ in Eq.(\ref{eq:trans-fix}). 
\end{proof}

%We conjecture, that (up to local relabeling of answers) the above game is the only one which satisfies the two properties properties above.

%Let us note here, that such a game for $d\geq 3$ is not linear \cite{review-nonloc}. Certain types of such games have been considered before \cite{??}, but to our knowledge, this class as a whole has not been studied yet in literature, though it seems to be natural generalization of XOR games (we will refer to them as GXOR). In what follows, we then focus on games satisfying the above relation.
The generalized XOR-d game is represented by a graph $G$ in analogous fashion to the XOR game. Namely, the vertices of the graph represent the inputs in the game and an edge between two vertices denotes that the corresponding measurements can be performed simultaneously. In the graph-theoretic representation of XOR-d games, we will use the notion of ``colors" to denote the edge-labelings that represent the winning constraints (permutations) in the game. We now also see the effect of the properties (P1) and (P2) characterizing the XOR-d game. While the graph-theoretic approach can also be applied to general unique games, most non-linear games have to be represented by a directed graph, as the permutations defining the winning constraints need not be their own inverse. 
%In a GXOR game all permutations assigned to edges satisfy $\pi=\pi^{-1}$ and thus the graph does not need to be directed.
%When talking specifically about GXOR games we will refer to the permutations assigned to the edges as "colors" to differentiate from other unique games. 
In Fig. \ref{fig:3colors}, we show an example of a game for a ternary output game with three possible winning permutations: red corresponding to $\pi_0$, blue corresponding to $\pi_1$ and green corresponding to $\pi_2$.
%according to the above definition of permutations (see Fig \ref{fig:3colors}).

\begin{center}
\begin{figure}[h!]
		\includegraphics[width=0.4\textwidth]{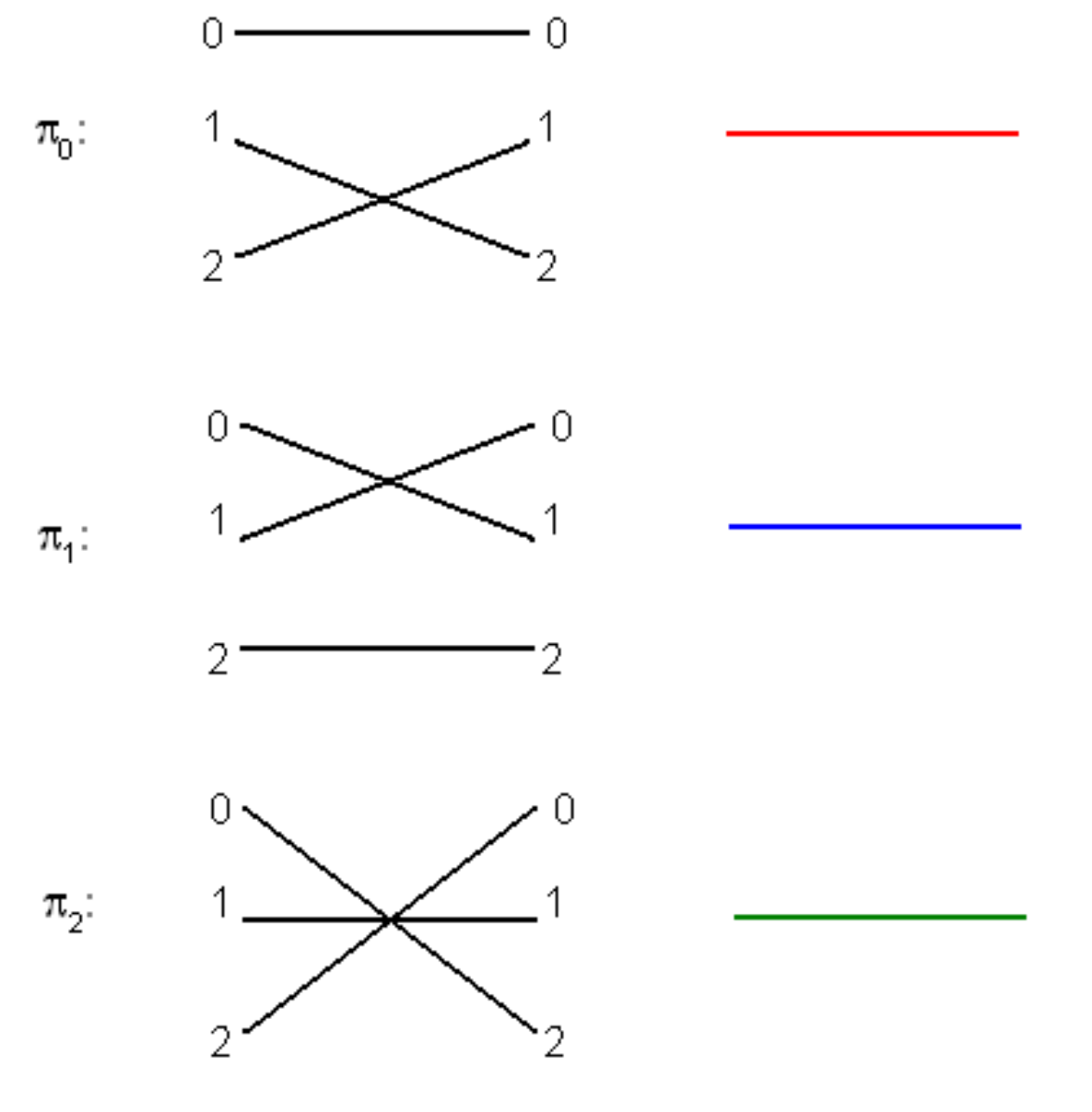}
\caption{(Color online) Definition of 3 permutations of a GXOR game with $d=3$ and corresponding colors. In later figures the colors will always denote the same permutations.}
	\label{fig:3colors}
\end{figure}
\end{center}

%\subsection{Single player linear non-contextuality inequality - a constraint satisfaction problem}
%\label{subsec:singleGENXORgame}
Note that the above formulation also naturally encompasses games with a single player, i.e., non-contextuality inequalities. In this case, the game scenario is simply a constraint satisfaction problem and is represented by a simple graph that is no longer constrained to be bipartite. The vertices still correspond to questions by the referee and the edge-labeling $K : E(G) \mapsto S_d$ denotes the permutations defining the winning constraints of the game. 
%One can depict the above mentioned game scenario as a constraint satisfaction problem, which is encoded as a simple graph, not necessarily bipartite, as in Fig. (\ref{fig:CHSH}). 
%With a game we associate a (simple) graph $G=(V(G),E(G))$ where $V(G)$ is the set of vertices and $E(G)$ is the set of edges. Vertices correspond to questions of the Referee, and the edges are labeled with permutations defining the constraints which the answers should satisfy. This labeling will be denoted by $K$:
%\be
%K : E(G) \mapsto S_d
%\ee 
%For example, in case of XOR game, there are only two types of labels: an edge can be solid $\pi_{sol}$ or dashed $\pi_{dash}$, denoting constraints for which answers in vertices connected by edge should be correlated or anticorrelated respectively. Hence the labeling in this case acts from $E(G)$ to the subset of 
%permutations $S_d$ which is $\{\pi_{sol},\pi_{dash}\}$.
The value of the single-player game (for a uniform input distribution) is simply,
% although in the case of a non-bipartite graph, we can no longer speak of "players":
\be
\omega_{S}(G,K) = \max_{P \in S} {1\over |E(G)|} \sum_{\{x,y\} \in E } V(a,b|x,y) P(a,b|x,y) 
\label{eq:def-of-pmax}
\ee
where $V(a,b|x,y) = 1$ iff $\pi_{(xy)}(a) = b$ and is $0$ otherwise, $|E|$ is the number of edges in the graph and $S$ is the classical, quantum or super-quantum set of boxes. It is worth noting that in the single-party scenario, a set of conditional probability distributions (box) $P(a,b|x,y)$ is {\it quantum} if it has the form $P(a,b|x,y) = Tr(\rho P_x^a Q_y^b)$ where $\rho$ is a quantum state and $P_x^a, Q_y^b$ are projection operators such that if $(x,y) \in E(G),$ then the commutator vanishes i.e., $[P_x^a, Q_y^b] = 0$. 
 
%This approach can be applied to any unique game. Most non-GXOR games, however, have to be represented by a directed graph, as the permutations used need not be their own inverse. In a GXOR game all permutations assigned to edges satisfy $\pi=\pi^{-1}$ and thus the graph does not need to be directed.

%We will now consider the classical, quantum and super-quantum value of the game. 
\paragraph{Classical value}
A well-known convexity argument shows that the optimal classical value of the game is obtained when the outcomes $a, b$ are assigned to the inputs $x, y$ in a deterministic manner. In terms of graphs, this can be formally described as follows.
Consider the assignment of deterministic values $f(x)$ in $\{0,...,d-1\}$ to each vertex $x$ of the graph $G$. If for some edge $e= (xy)$ of $G$ one has $\pi_e (f(x)) \neq f(y)$ i.e., if the values of the assignment do not satisfy the winning constraint defined by the color (permutation) associated with the edge, we say that there is a {\it contradiction}. Then the minimal number of contradictions over all deterministic vertex assignments for a graph $G$ with a given edge-labeling $K:E(G)\mapsto S_d$ is denoted as $\beta_C(G,K)$. This quantity characterizes the classical value of the game:
\be 
\omega_{C}(G,K) = 1-\frac{\beta_C(G,K)}{\left|E(G)\right|}
\ee

%Let us recall here, that a set of conditional probability distributions (box) $P(a,b|x,y)$ is {\it quantum} if it has the form $P(a,b|x,y) = Tr(\rho P_x^a Q_y^b)$ where $\rho$ is a quantum state and $P_x^a, Q_y^b$ are projection operators such that if $(x,y) \in E(G),$ then $P_x^a Q_y^b=Q_y^b P_x^a$. 
\paragraph{Super-Quantum Value}
Super-quantum is a set of all conditional probability distributions (referred to also as behaviors or boxes) $P(a,b|x,y)$ which are {\it consistent}, i.e., they satisfy the criterion that the marginal distribution $P(a|x)$ is consistenly defined for each vertex $x$ of the graph in a manner independent of the context (clique of the graph) in which it appears.  
%
%{\definition For a given graph $G=(V,E)$, $P(a,b|x,y)$ is a {\it consistent} box if for all pairs of cliques $V',V'' \subset V$, and for set of vertices $W =  V' \cap V'' \neq \emptyset$ there is
%\ben
%\forall_{w \in \{0,1\}^{|W|}},\,\, 
%\sum_{u' \in \{0,1\}^{|U'|}} Pr(W=w,T=t|v=V') = \nonumber \\
%= \sum_{u'' \in \{0,1\}^{|U''|}} Pr(W=w,T'=t'|v=V'')
%\een
%where $U'= V' \setminus W$ and $U'' = V'' \setminus W$.
%\label{def:consistent}
%}
In the case of a bipartite graph $G$ with the bi-partition of $V(G)$ being the set of inputs of the two parties, the above consistency condition is nothing but the no-signaling condition given in Eq (\ref{eq:nscondition}). 
%It is known, that the upper bound on the quantum value can be obtained using the method of the so called NPA hierarchies \cite{navascues-2008-10}. In the case of GXOR games with more than two colors we use the $h_{3/2}$ hierarchy. Numerical results on that can be found in Section \ref{subsec:2and3colors}. 
%It also is known, that the set of super-quantum behaviors is strictly larger
%than that of quantum behaviors. It can be easily seen by studying the attainable value of game by super-quantum behaviors.
With super-quantum resources, for any graph $G$ and edge labeling $K$, one readily gets that $\omega_{SQ}(G,K) = 1$. To see this, consider a behavior $P(a,b|x,y)$ satisfying
\be
\label{eq:sq-win-strat}
P(a,b|x,y)  = P(a,\pi_{e} (a)) = {1\over d}
\ee
for all edges $e= (x, y) \in E(G)$ i.e. the maximally correlated distribution (according to the permutation $\pi_e$) over all outcomes at the edge.
%permuted according to $\pi_e$. 
Then, by definition all constraints are satisfied with probability $1$, hence $\omega_{SQ}(G,K) = 1$ as desired. Moreover, since the marginal distribution at each vertex for the above strategy is simply given by $P(a|x) = \frac{1}{d}$, we have that Eq. (\ref{eq:sq-win-strat}) is a well-defined super-quantum box. 

\subsection{XOR-d games for partial functions}
\label{subsec:GOXR-uncolored}
In this section, we consider the possibility of XOR-d games corresponding to partial functions $f(x,y)$, i.e., where the winning constraints are only defined for a subset of input pairs $(x,y)$. We incorporate this in the graph-theoretic formulation by simply allowing the edge-labeling to leave some edges uncolored. However, since the measurements corresponding to the two vertices in the uncolored edge might still be required to commute, we depict these as gray edges. An important example where such edges naturally arise is the Braunstein-Caves Chained Bell inequality \cite{BraunsteinC1988}. This inequality concerns a game with $N^2$ inputs which has numbers from the set $\{0,2,...,2N-2\}$ for Alice and from the set $\{1,3,...,2N-1\}$ for Bob. However, the winning constraints are only defined for $2N$ neighboring pairs $\{(k,k+1 \; \text{mod 2 N}) : k \in \{0,...,2N-1\} \}$ and only these enter the chained Bell expression. The corresponding graph has $2N$ edges forming a cycle. But all of Alice's measurements commute with all of Bob's measurements, so that the additional gray edges are added. This distinguishes the chained Bell inequality in the two-party scenario from the $2N$ cycle contextuality game \cite{AQBCC13} which is simply depicted by the cyclic graph $C_{2N}$. 
%
%one has to take into account also, that each Alice's measurement commutes with that of Bob, which makes the need of additional type of edge - the one which reports commutation. In such a way, we make a distinction between the Chained Bell inequality which contains additional gray edges between all Alices
%and all Bob's vertices and the contextuality "game" \cite{AQBCC13} which is depicted as cycle.

%We consider a variation of GXOR games in which some edges of the graph $G$ may be left uncolored. This can be formally defined 
In the partial function XOR-d game, we have a sub-graph $G' = (V, E')$ of $G$ with a labeling $K':E'\mapsto L_d,$ where $E'\subset E(G)$. The gray edges $(x,y) \in E(G)$ denote that the observables represented by the vertices $x$ and $y$ must commute, but they do not have to satisfy any other constraints. The success probability in the game is thus given as
\be
\omega_{S}(G',K') = \max_{P \in S} {1\over |E'|} \sum_{ (x,y) \in E' } V(a,b|x,y) P(a,b|x,y) 
\ee
Clearly, in the classical case the minimum number of contradictions for a given $G=(V,E)$ and $K:E'\mapsto L_d$ is equal to $\beta_C(G',K)$ and thus $\omega_{win}(G,K)=\omega_{win}(G',K).$ This is not necessarily true for the quantum case, since vertices connected by a gray edge still have to commute. Nevertheless, we have the following straightforward general dependencies.
If $K:E\mapsto L_d$ is any edge-labeling of $G$ such that $K(e)=K'(e)$ for any edge $e\notin E',$ the following inequalities are true:
\ben 
\gamma_{C}(G,K)-\left|E-E'\right| &\leq & \gamma_{C}(G,K')=\gamma_{C}(G',K'),\, \nonumber \\
\gamma_{Q}(G,K')-\left|E-E'\right| & \leq & \gamma_{Q}(G,K)\leq \gamma_{Q}(G',K), \nonumber \\
\een
where $\gamma_{C}$ and $\gamma_{Q}$ denote un-normalized classical and quantum values, that is $\omega_{S}(G,K)\left|E(G)\right|$ with $\omega_S(G,K)$ defined in equation (\ref{eq:def-of-pmax}), respectively.

\begin{figure}[h]
	\includegraphics[width=0.40\textwidth]{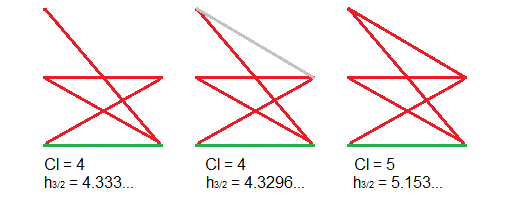}
	\caption{Example: Un-normalized classical and $\gamma_{3/2}$ values for graphs (a) $(G',K')$ - without the edge $e$, (b) $(G,K')$ - with an uncolored (gray) edge $e$ and (c) $(G,K)$ - with a colored $e$. The graphs (a), (b) show that while gray edges do not affect the classical value, they can potentially affect the quantum value of a game. }
	\label{fig:legs}
\end{figure}

\section{Equivalent Games}
%\section{Labeled graph equivalence}
\label{sec:labeled-eq}
%\com{signle partite or contextual or constraint sat. problem ?}
In this section, we use the graph-theoretic approach to find non-equivalent games both in the single- and two-party scenario. Two games are equivalent when they can be transformed into each other by operations which do not change the winning probability, i.e., $\omega_c(G)$ and $\omega_q(G)$ are equal for these games. The operation transforms the edge labeling of one game graph into that of the other. 
%This is when the operation transformes properly colors of edges in one graph into colors of edges 
%of the other graph, that is transformes labeling of one graph into labeling of the other. (In analogy two equivalent versions of the CHSH game that has been depicted in previous section on Fig \ref{fig:CHSH}).
%In the case of XOR games these equivalences are characterized via the notion of {\it signed graphs}. We first explain this approach and then consider equivalent XOR-d games. 
%In first subsection we explain this approach, and further pass to explain operations which allow to find equivalences in case of GXOR games.

\subsection{Equivalent XOR games using signed graphs}
The XOR game graphs are in fact equivalent to the well-known class of signed graphs \cite{Har}, i.e., graphs with 'positive' and 'negative' edges.  
%Such a graph $G$ with edge-labeling $K:E(G)\mapsto \{\id,(01)\}$, and thus the corresponding game, can also be represented as a signed graph \cite{Har}, i.e. a graph with 'positive' and 'negative' edges. 
Positive edges correspond to edges labeled with identity (correlations) and negative to edges labeled with $(01)$ (anti-correlations). Signed graphs are much studied in literature due to their extensive use in modeling social processes \cite{Roberts} and also because of their interesting connections with classical mathematical systems \cite{Zaslavsky}. A cycle in a signed graph is said to be \textit{balanced} if it contains an even number of negative edges, a signed graph itself is said to be balanced if all of its cycles are balanced. 

%Two signed graphs are considered equivalent if they have the same underlying (unsigned) graph $G$ and identical sets of balanced cycles. Clearly, all signed graphs equivalent to a balanced graph must also be balanced and all graphs equivalent to an unbalanced graph are also unbalanced. Furthermore, labeled graphs $(G_1,K_1)$ and $(G_2,K_2)$ corresponding to equivalent signed graphs contain the same number of contradictions. Thus the games corresponding to equivalent signed graphs are also equivalent. 

A marking of a signed graph is a function $\mu: V(G) \rightarrow \{+,-\}$. Switching $(G, K)$ with respect to a marking $\mu$ is the operation of changing the sign of every edge label of $G$ to its opposite whenever its end vertices are of opposite signs.
Formally, we have that equivalent XOR games correspond to switching equivalent signed graphs. Switching equivalent signed graphs $(G_1,K_1)$ and $(G_2,K_2)$ are cycle isomorphic , i.e., there exists an isomorphism $\phi: G_1 \rightarrow G_2$ such that the sign of every cycle $Z$ in $(G_1, K_1)$ equals the sign of $\phi(Z)$ in $(G_2, K_2)$ \cite{Zaslavsky}. 
%Two signed graphs are equivalent if and only if one can be obtained from the other by isomorphism and switching the signs of all edges incident to certain vertices. 

\subsection{Equivalent XOR-d games: Labeled graph equivalence}

We now generalize the notion of signed graph equivalence from \cite{Har} to find equivalent XOR-d games both in the single- and two-party scenarios. 
%We can define an equivalence between labeled graphs, and thus the corresponding unique games, by generalizing signed graph equivalence (\cite{Har}). 
We consider two labeled graphs $(G_1,K_1)$ and $(G_2,K_2)$ to be equivalent if one can be obtained from the other by isomorphism between $G_1$ and $G_2$ and switching operations $s(v,\sigma)$ which we define below. In terms of games, switchings correspond to local operations such as relabeling of outputs by the players. 
%Thus, they preserve all the relevant properties of the game. We introduce them now formally.

For any graph $G$ and edge-labeling $K:E(G)\mapsto S_d$, let $v\in V(G)$ be any vertex of $G$ and let $\sigma$ be any permutation of $[d].$ For every edge $e$ incident to $v$ we change color (i.e. permutation) $\pi$ of the edge $e$ into $\pi\sigma$ where $\sigma$ is some permutation, which we will specify later.
Such a change defines a new edge-labeling $\hat{K}$ as follows:
\ben
\forall e \in E(v)\,\, \hat{K}(e)= \pi\sigma \nonumber\\
\forall e \notin E(v)\,\, \hat{K} \equiv K.
\een
where $E(v)$ is the set of edges incident with vertex $v$, that is of the form $e=\{v,u\}$ for some $u\in E(G)$.

\begin{figure}[h]
	\includegraphics[width=0.40\textwidth]{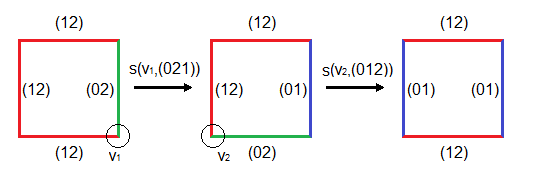}
	\caption{The third graph is obtained from the first by two switches. Thus, they are equivalent.}
	\label{fig:switch2}
\end{figure}

Note that the above operation applies not only to XOR-d games, but to all unique games. 
In fact, labeled graphs representing some non-linear unique games may be equivalent to some XOR-d games. 
If we wish to obtain only XOR-d games equivalent to a given XOR-d game, we have to limit the permutations $\sigma$ used in the switching operations to the set of such permutations $\sigma\in S_d$ that $\sigma\pi\in L_d$ for all $\pi\in L_d$. Since for every $\pi_i,\pi_j\in L_d$ there exists a permutation $\sigma_k\in L_d'=\{\sigma_i:\sigma_i(x)=i+x \mod d\}$ such that $\sigma_k\pi_i=\pi_j,$ we can obtain all XOR-d games equivalent to a given XOR-d game using only switches with permutations from the set $L_d'.$ In fact, since every $\sigma_i\in L_d'$ is equal to $\sigma_1^i,$ where $\sigma_1(x)=x+1,$ we only need to consider $\sigma_1$ multiple times for the same vertex. For example in the case of a XOR-d game with $d=3,$ i.e., a graph labeled with three colors (i.e. $L_3=\{(12),(02),(01)\}\subset S_3$) all XOR-3 games equivalent to it can be obtained via $s(v,(012))$ applied (multiple times) for each $v\in V$, and their automorphic copies. The above notion can also be extended to include graphs with uncolored edges. In this case the switching operation $s(v,\sigma)$ changes the color $K(e)=\pi$ with $\sigma\pi$ for all colored edges incident to $v$, while uncolored edges remain unaffected. It is easy to see that the equivalence still preserves all relevant properties of the game.

\section{Classical value}
\label{sec:bin-xor}

The graph theoretic approach is also useful for studying the classical values of XOR, XOR-d and other unique games. 
%I also allows us to find equivalent games which have the same classical and quantum values.
As we have seen, the classical value of an XOR-d game (for both Bell and non-contextuality inequalities) defined by a graph $G$ with edge-labeling $K$ obeys
%A Bell/contextual inequality defined by $G$ and $K$ has the form  
\be 
\omega_{C}(G,K) = 1-\frac{\beta_C(G,K)}{\left|E(G)\right|}
\ee
where $\beta_C$ denotes the minimum number of contradictions over all deterministic vertex-assignments. To study this number we will use, in particular, a graph constructed from the graph $G$ and edge-coloring $K$ which we simply call $KG$.

\subsection{XOR games}
%\subsubsection{One and two party XOR games}
We can characterize the contradiction number (and hence the classical value) of a general XOR game (with one or two parties) in a graph-theoretic manner as follows.
\begin{theorem}
$\beta_C(G,K)$ is equal to the minimal number of edges which need to be removed from $G$ so that the resulting graph does not contain any cycle with an odd number of dashed $(01)$ edges.
\end {theorem}

Expressed in terms of labeled graphs, this states that a graph $G$ with edge-labeling $K:E(G)\mapsto S_2$ has a consistent vertex-assignment if and only if it has no cycles with an odd number of edges labeled with $(01)$. Thus, $\beta_C(G,K)=0$ if and only if there are no such cycles in the graph. The problem of calculating the classical value of a XOR game, and $\beta_C(G,K),$ is known to be NP-hard \cite{Hastad}. The proof of the statement follows directly from the fact that every unbalanced cycle leads to a contradiction, and from the following characterization of balanced signed graphs in \cite{Har}.

\begin{fact}[\cite{Har}]
A signed graph is balanced if and only if its set of vertices can be partitioned into two disjoint subsets in such a way that each positive edge joins two vertices in the same subset while each negative edge joins two vertices from different subsets.
\end{fact}

\subsubsection{Complexity of computing the classical value for single color XOR games}
\label{subsec:single-color-XOR}

We now consider a subclass of XOR games in which the winning constraints only ask for anti-correlations between the outcomes. This type of game is represented by a graph in which all the edges are dashed (i.e. labeled by the permutation $\pi=(01)$). Clearly, all bipartite graphs with such a labeling are satisfiable, i.e., the corresponding Bell inequalities have classical value one. Thus, single color games are trivial in the Bell scenario and only relevant in a scenario of contextual games. Also, for general graphs if the edges are all solid (labeled by the identity) then clearly, the game is won by a classical strategy. We now characterize the classical value of contextuality games corresponding to single $(01)$ color non-bipartite graphs, as we shall see computing the classical value is hard even in this simplest possible scenario. 
%
%
%We now move on to study non-bipartite graphs with a single color. Observe first, that if all edges are solid, then clearly, the graph is satisfiable. It is however not the case if the graph contains only dashed edges. We shall now characterize the classical value of single color games.

\begin{obs} For a graph $G$ with dashed edges only, $\beta_C(G,K)$ equals the minimal number of edges needed to be removed, so that the resulting graph is bipartite.
\end{obs}
\begin{proof}Clearly, a bipartite graph with only dashed edges is satisfiable: one can assign value $0$ to all vertices in one 
partition, and value $1$ to the vertices in the other partition. To see the converse, recall that a graph is bipartite if and only if it does not contain an odd cycle. Now, if a graph $G'$ obtained from $G$ by removal of edges is not bipartite, it must contain an odd cycle. An odd cycle of $(01)$ edges clearly contains a contradiction for every vertex assignment. 
%
%
%which is $G$ (after removing some edges) is not yet bipartite, then it has to contain an odd cycle. Such a cycle however certainly contains a contradiction for any vertex assignment, hence has to be removed, or else $G'$ is not satisfiable. This proves the thesis.
\end{proof}

%It follows that all bipartite single color graphs are satisfiable. Thus, single color games are trivial in the Bell scenario and only relevant in a scenario of contextual games or in general contextual inequalities, as they involve single party \com{contextual -check the meaning in all file}.
Thus, determining the classical value of a single color contextuality XOR game is equivalent to finding the edge-bipartization number $\beta_c^{(2)}$ of the corresponding graph. This problem is known to be MaxSNP-hard \cite{PY91}. It can be approximated to a factor of O$(\sqrt{\log n})$ in polynomial time, where $n$ is the total number of vertices (see \cite{ACMM05}). Also, note that assuming the Unique Games Conjecture, it is NP-hard to approximate Edge Bipartization within any constant factor \cite{Khot02}. 

Note that for the corresponding single color subclass for XOR-d games, the edge-bipartization number only gives an upper bound on $\beta_C(G, K)$ (a lower bound on the classical value).  
%It should be noted that a graph labeled with only one out of $d>2$ colors (i.e. subclass of XOR-d games) is an entirely different problem. 
Since all cycles of even length in such graphs have $\beta_C=0$, removing all cycles of odd length will result in a graph without contradictions. 
%Thus the edge bipartization number is an upper bound on $\beta_C(G,K).$ 
However, this is not always an optimal solution. For example, considering $C_5$ (the cycle graph of length $5$) labeled with any single permutation $\pi \in L_3$, we find that $\beta_C^{(2)} = 1$ while clearly $\beta_C = 0$ since a vertex assignment satisfying such a winning constraint can always be found (by assigning the same value to all vertices according to $\pi$). 

% the cycle of length 5, $C_5$ if labeled with any single permutation $\pi_i\in L_3$ has $\beta_c^2=1$ and $\beta_C=0$. 

\subsection{XOR-d games}

The classical value of the specific XOR-d game called the CHSH-d game has been studied in \cite{BavarianShor, Pivoluska} using techniques from algebraic geometry. In this section we study the classical value of  generalized XOR-d games using graph-theoretic methods. 
%We first try to apply the approach which succeeded in the case of XOR games, as shown in previous section in order to characterize the latter value in graph-theoretical terms.
%section \ref{sec:bin-xor}
Clearly, if the game graph is cycle-free (forms a tree), then any set of winning constraints for this graph can be satisfied. Hence it must be the presence of the cycles, which disallows satisfiability. Just like an unbalanced cycle in an XOR game graph leads to a contradiction, there are also "bad" cycles in XOR-d game graphs. These are the cycles for which no consistent vertex-assignment exists that satisfies all the winning constraints in the cycle. There are "good" cycles in an XOR-d game graph analogous to the balanced cycles in the binary XOR case, for which any consistent vertex assignment satisfying the winning constraint is admissible.   However, in the case of XOR-d game graphs, we encounter new "ugly" cycles, for which only certain particular vertex assignments satisfy the cycle (see Fig \ref{fig:good-bad-ugly}). It then becomes a non-trivial question to study how many edges one needs to remove in order to make a graph satisfiable, as for instance in the Fig. \ref{fig:przyklad} removing a bridge (a single edge connecting two components) of the graph can lead to a better result than a brute-force removal of one edge per each "ugly" cycle. 
%
%Studying XOR games, we distinguish a type of cycles that are so to speak "bad", and they characterize the problem. Moving to GXOR games, one can think that analogously, removing minimal number of edges so that a kind of "bad" cycles are absent in the graph, should lead to minimal number of unsatisfied constraints in the classical case. However, as we show, apart from "bad" cycles, which have no consistent vertex-assignment and certainly should be removed, and "good", that can be left in a graph, as for them all vertex-assignments satisfy the constraints, there are also "ugly" ones that are neither: they some vertex-assignments satisfy the constraints, but not all assignments do as it is in case of the "good" cycles. For this reason it may be impossible to satify all of them at once (see Fig \ref{fig:good-bad-ugly}). 
\begin{center}
\begin{figure}[h!]
		\includegraphics[width=0.4\textwidth]{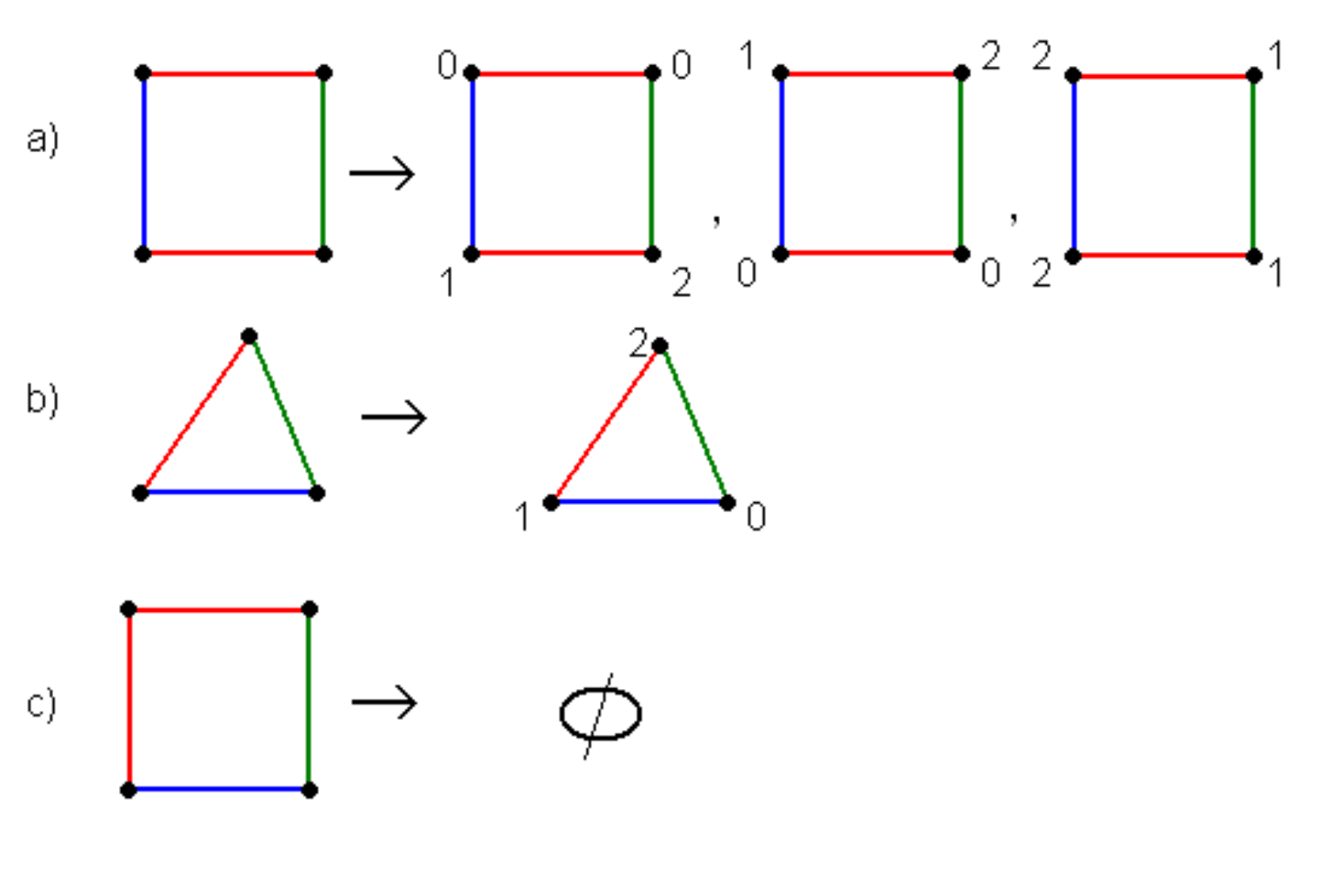}
\caption{(Color online)(a) Exemplary good cycles - with all 3 consistent assignments (b) ugly cycles with only 1 consistent assignment and (c) bad cycles with no consistent assignment. }
	\label{fig:good-bad-ugly}
\end{figure}
\end{center}

%This makes characterization a highly non-trivial problem.

\begin{center}
\begin{figure}[h!]
		\includegraphics[width=0.4\textwidth]{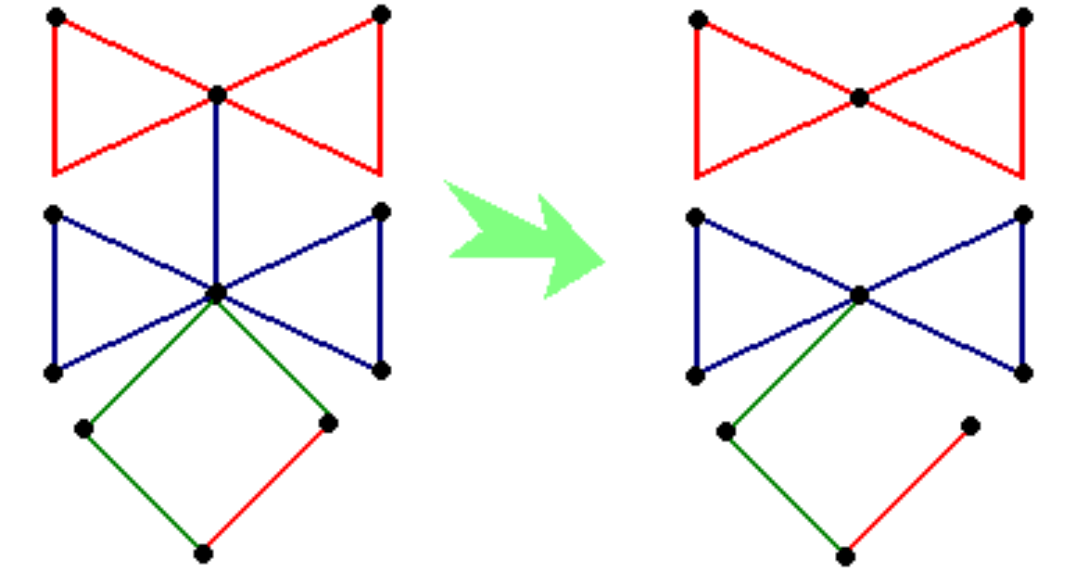}
\caption{(Color online) Removing the blue edge which is  a bridge and single green edge, leaves two components that are already satisfiable. It is then cheaper in terms of edges than removing three edges in a way that both the red component and the green one are cycle-less.}
	\label{fig:przyklad}
\end{figure}
\end{center}

%\com{here was formal GXOR again, moved to app, I think not needed }

\subsubsection{The Good, bad and the ugly cycles}

We say that a cycle $C$ in a graph $G$ with edge-labeling $K:E\mapsto S_d$ is \textit{bad} if it has no vertex-assignment that satisfies the constraints and \textit{good} if it has $d$ such assignments (i.e. the largest possible number), otherwise the cycle is \textit{ugly}. We denote by $\xi$ with the corresponding subscript (g,b,u) the number of good, bad and ugly cycles respectively. 
%The number of cycles that are bad (good, ugly) we will denote by $\xi$ with subscript denoting the set of cycles e.g. $\xi_{bad}$ for number of bad cycles.
% Guided by our intuition from previous section, we first characterise the number of contradictions in cycles. This will lead us to three classes of cycles announced above: {\it good, bad} and {\it ugly}.
Clearly, any bad cycle has to be removed to make the graph satisfiable, while also removing all the ugly cycles necessarily leaves a satisfiable graph. Now, if there are no ugly cycles $(\xi_{u}(G,K) =0)$, then $\beta_C(G,K) =\xi_{b}(G,K)$ if however 
$\xi_{u}(G,K) >0$, we can leave at least one ugly cycle. This is because the single ugly subgraph has an assignment, which determines a consistent assignment for the whole graph. Therefore, we have the following observation. 
\begin{obs}
For any XOR-d game graph $G$ with edge-labeling $K$ using $d$ colors, if $\xi_{u}(G, K) = 0$ then $\beta_C(G,K) =\xi_{b}(G,K)$, and if $\xi_{u}(G, K) > 0$  we have
\be
\xi_{b}(G,K) \leq \beta_C(G,K) \leq \xi_{b}(G,K) + \xi_{u}(G,K) -1
\ee
%\be
%\xi_{bad}(G,K) \leq \beta_C(G,K) \leq \xi_{bad}(G,K) + \xi_{ugly}(G,K)
%\ee
\end{obs}
%
%
%Note, that if there are no ugly cycles $(\xi_{ugly}(G,K) =0)$, then $\beta_C(G,K) =\xi_{bad}(G,K)$ if however 
%$\xi_{ugly}(G,K) >0$, we can leave at least one ugly cycle in a graph  i.e. 
%\be
%\xi_{bad}(G,K) \leq \beta_C(G,K) \leq \xi_{bad}(G,K) + \xi_{ugly}(G,K) -1
%\ee
%This is because the single ugly subgraph has an assignment, which determines a consistent assignment for the whole graph.

It is clear that a graph with only one cycle can have at most one contradiction. Whether or not there is a contradiction can be determined through the composition of all permutations assigned to the cycle's edges. We define a permutation $\pi_{C_t}=K(e_1)K(e_2)...K(e_t),$ where $e_i\cap e_{i+1}=\{v_i\}$ for all $i$ and $v_t=v_0.$

\begin{theorem}
A cycle $C_t$ has a consistent vertex-assignment for a given edge-labeling $K$ if and only if $\pi_{C_t}$ has at least one fixed point.
\end{theorem}

\begin{proof} 
%{\it Proof}.- 
It is easy to see, that for a vertex-assignment $k:V(C_t)\mapsto [d]$ a contradiction happens in $C_t$ iff there exists $k(v_0) \neq \pi_c(k(v_0))$
where $\pi_c \equiv K(e_{t-1})K(e_{t-2})...K(e_1)$ with $e_1$ an edge incident with $v_0$, and $v_t=v_0$. This hawever is equivalent to the fact
that $k(v_0)$ is a fixed point of $\pi_c$.
%Suppose $k:V(C)\mapsto [d]$ is a consistent vertex-assignment of $C.$ For any $v_i\in C,$ $k(v_{i+1})=K(e_{i+1})(k(v_i))=K(e_1)K(e_2)...K(e_{i+1})(k(v_0)).$ Since $v_d=v_0,$ and therefore $k(v_t)=k(v_0),$ we have $\pi_c(k(v_0))=k(v_0).$  
\end{proof}

\begin{cor}
The number of fixed points of $\pi_{C_t}$ is equal to the number of consistent vertex-assignments of $C_t.$
\end{cor}

It follows that the number of contradictions in a given graph is at most the number of cycles. It may, however, be greater than the number of bad cycles. 

\subsubsection{The graph KG}
%\kh{For completeness of the presentation we invoke now some results which will be presented in \cite{RS}.}
To study the number of contradictions and consistent vertex-assignments in a given graph $G$ with edge-labeling $K:E(G)\mapsto S_d$, we define the graph $KG$, described in more detail in \cite{RS}. This graph is constructed as follows.
\begin{enumerate}
\item Replace each vertex $v_i\in V(G)$ with a disjoint set $\{v_{i0},...,v_{in-1}\}\in V(KG)$ of $d$ vertices. 

\item Connect two vertices $v_{is}, v_{jt}\in V(KG)$ with an edge if and only if the graph $G$ has an edge $v_iv_j$ and $\pi_{ij}(s)=t,$ where $\pi_{ij}=K(v_iv_j).$
\end{enumerate}

For a connected graph $G$ the {\it assignment number} $\beta'_C(G,K)$ is equal to the number of connected components of $KG$ isomorphic to $G.$ Each such component contains exactly one vertex from the set corresponding to a given vertex. Thus, a consistent vertex-assignment exists if and only if there exists a vertex $v_i$ not connected to any $v_j\in\{v_0,...,v_d\}.$ 

\begin{theorem} \cite{RS}
For any given $G_1, G_2, K_1:E(G_1)\mapsto S_d, K_2:E(G_2)\mapsto S_d$ the labeled graphs $(G_1,K_1)$ and $(G_2,K_2)$ are equivalent if and only if $K_1G_1$ and $K_2G_2$ are isomorphic.
\end{theorem}

\begin{figure}[h]
	\includegraphics[width=0.40\textwidth]{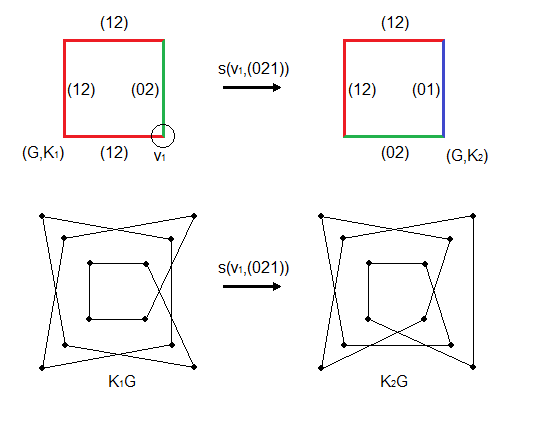}
	\caption{A switching operation on (G,K) and the corresponding isomorphism of KG.}
	\label{fig:switch1}
\end{figure}

It follows that the contradiction numbers $\beta_C(G,K)$ and $\beta_C(G',K')$ of two equivalent labeled graphs $G, G'$ are the same. Analogously these
graphs have the same $\beta_C'(G,K)$. This fact holds true even for some non-linear, but unique games. If $G$ is a bipartite graph and $K:E\mapsto L_d$ (i.e., a XOR-d game) every cycle in $G$ has either $0$ or $d $ consistent vertex-assignments. Furthermore, in \cite{RS} the following theorem is proved:

\begin{theorem}\cite{RS}
For any edge-labeling $K:E\mapsto L_d$ a complete bipartite graph $K_{s,t}$ (i) has no ugly cycles and (ii) is bad if and only if it contains a bad cycle of length 4.
\label{thm:4-cycles}
\end{theorem}

We will now consider a type of game in which each of the two players has $d$ possible answers. 
This game corresponds to the complete bipartite graph $K_{s,t}$ with an edge-labeling $K:E\mapsto L_d.$ Thus, to find the classical bounds we search for the minimal set of edges which need to be deleted so that there are no more induced cycles with contradictions. In the case of $K_{3,3}$ and smaller bipartite graphs, by theorem \ref{thm:4-cycles} to make it good, we need only to delete edges until all remaining cycles of length $4$ are good. For all possible edge-labelings of $K_{3,3},$ with three colors $0\leq\beta_C(G,K)\equiv \beta_C\leq 3$. For about $1.23\%$ of labelings $K$, there is $\beta_C=0,$ in $22.22\%$ of cases $\beta_C=1,$ in $74.07\%$ of cases $\beta_C=2$ and in $2.5\%$ of cases $\beta_C=3.$  

%\begin{center}
%\begin{figure}[h!]
%		\includegraphics[width=0.4\textwidth]{wykres.pdf}
%\caption{Out of all labelings of $K_{3,3}$ with three colors: $1.23\%$ have no contradictions, in $22.22\%$ of cases $\beta_C=1,$ in $74.07\%$ of cases $\beta_C=2$ and in $2.5\%$ of cases $\beta_C=3.$}
%\end{figure}
%\end{center}

\section{Quantum Value : Lov\'{a}sz theta as an upper bound for a single-party contextuality game}
\label{sec:Lovas}

For two-party XOR games, the theorem of Tsirelson \cite{Tsirelson} and the subsequent analysis in \cite{Cleve} gives an efficient semi-definite programming method to compute the exact quantum value. For general XOR-d games however, this is no longer the case and the semi-definite programming hierarchy of \cite{navascues-2008-10} has to be applied. It is at present unknown whether the quantum value of these Bell inequalities can be obtained at some particular level of the hierarchy. An efficiently computable upper bound on the quantum value of general XOR-d and other linear game Bell inequalities was proposed in \cite{RAM} and subsequently generalized to the multi-party scenario in \cite{MRMC15}. 

For single-party contextuality, in \cite{CSW}, it was shown that the quantum value of any non-contextuality inequality involving projectors represented in an orthogonality graph $\Gamma$ is given by the (weighted) Lov\'{a}sz theta number $\theta_w(\Gamma)$ of the orthogonality graph. Analogously, the classical value of the inequality is given by the (weighted) independence number $\alpha_w(\Gamma)$ of the orthogonality graph. While calculating the independence number of an arbitrary graph is a well-known NP hard problem, calculating the Lov\'{a}sz theta number can be achieved by means of a semi-definite program. As such, in the scenario of single-party contextuality as studied in the traditional ``Kochen-Specker" scenario \cite{Kochen-Specker} involving yes-no questions represented by projectors in quantum theory, the quantum value was exactly and efficiently computable by an SDP. Therefore, for the single-party XOR games and their generalization to XOR-d studied so far, one might wonder whether the quantum value is still efficiently computable. The answer to this question turns out to be negative even in the single party scenario.
% when one considers general non-contextuality inequalities. 

Let us first describe for a given single-party contextuality game represented by a commutation graph $G$, the method of constructing the corresponding orthogonality graph $\Gamma$, from which we might hope to calculate the quantum value. 

\begin{itemize}
\item Firstly, we list all the maximal cliques $\{C_1(G), \dots, C_m(G)\}$ of the commutation graph $G$, where a maximal clique refers to a complete subgraph that cannot be enlarged. Each maximal clique corresponds to a set of $d$-outcome observables $\{A_i^{(j)}(G)\}$, i.e., $C_i(G) = \{A_i^{(1)}(G), \dots, A_{i}^{(k)}(G)\}$  where $k \leq \omega(G)$ with $\omega(G)$ being the clique number of the commutation graph $G$. 

\item For each maximal clique $C_i(G)$ of size $k$ we list a set of $d^k$ vertices of a new orthogonality graph $\Gamma$. Each of the $d^k$ vertices of $\tilde{C}_i(\Gamma)$ corresponds to an event of the form $(l_1, \dots, l_k|A_1,\dots, A_k)$ with associated projector $\otimes_{j=1}^{k} \Pi_{A_{i}^{(j)}}^{l_{j}}$ for $l_j \in \{1, \dots, d\}$. 

\item Two vertices in $\Gamma$ are connected by an edge if the corresponding projectors are locally orthogonal. In other words, for vertices $u$ and $v$ corresponding to events $(l_1(u), \dots, l_{k_1}(u) |A_1(u), \dots, A_{k_1}(u))$ and $(l_1(v), \dots, l_{k_2}(v) |A_1(v), \dots, A_{k_2}(v))$ are connected by an edge $u \sim v$ if $\exists j_1 \in [k_1], j_2 \in [k_2]$ such that $A_{j_1}(u) = A_{j_2}(v)$ and $l_{j_1}(u) \neq l_{j_2}(v)$.  
%projectors $\otimes_{j_u=1}^{k_u} \Pi_{A_{i_u}^{(j_u)}}^{l_{j_u}(u)}$ and $\otimes_{j_v=1}^{k_v} \Pi_{A_{i_v}^{(j_v)}}^{l_{j_v}(v)}$ %respectively, we have that $u \sim v$ if $\exists j_u \in [k_u], j_v \in k_v$ such that $A_{i_u}^{(j_u)} = A_{i_v}^{(j_v)}$ and $l_{j_u}(u) %\neq l_{j_v}(v)$. 
We thus see that each maximal clique $C_i(G)$ of size $k$ of the commutation graph $G$ corresponds to a $d^k$ sized maximal clique $\tilde{C}_i(\Gamma)$ of the orthogonality graph $\Gamma$. 
%of the corresponding orthogonality graph $\tilde{\Gamma}$.

\end{itemize}

Each of the probabilities $P(a,b = \Pi_{x,y}(a)|A_x, A_y)$ appearing in the game expression can be expressed (as marginals) in terms of the probabilities $P(l_1, \dots, l_k | A_1, \dots, A_k)$, so that the game expression can be written as a weighted sum of probabilities of the events appearing in the graph $\Gamma$. An orthonormal representation of a graph $\Gamma$ is a set of unit vectors $|u_v \rangle$ (with $\| |u_v\rangle \| = 1$) such that for $v_1 \sim v_2$ we have $\langle u_{v_1} | u_{v_2} \rangle = 0$. The weighted Lov\'{a}sz theta number of the graph $\Gamma$ was defined by Lov\'{a}sz as \cite{Lovasz-0}
\be 
\theta_w(\Gamma) = \max_{|\psi \rangle, \{|u_v \rangle\}} \sum_{v \in V(\Gamma)} w_v |\langle \psi | u_v \rangle |^2
\ee
where the maximum is over orthonormal representations $\{|u_v \rangle \}$ of $\Gamma$ and an arbitrary normalized unit vector $|\psi \rangle$. Here $V(\Gamma)$ denotes the set of vertices of the graph and $w_v$ denotes the weight with which the probability $P(l_1(v), \dots, l_{k_1}(v) |A_1(v), \dots, A_{k_1}(v))$ associated to the vertex $v$ enters the game expression. 
An example of a commutation graph $G$ and its corresponding orthogonality graph $\Gamma$ is shown in Figure \ref{fig:com-ort}. 
\begin{center}
\begin{figure}[h!]
		\includegraphics[width=0.5\textwidth]{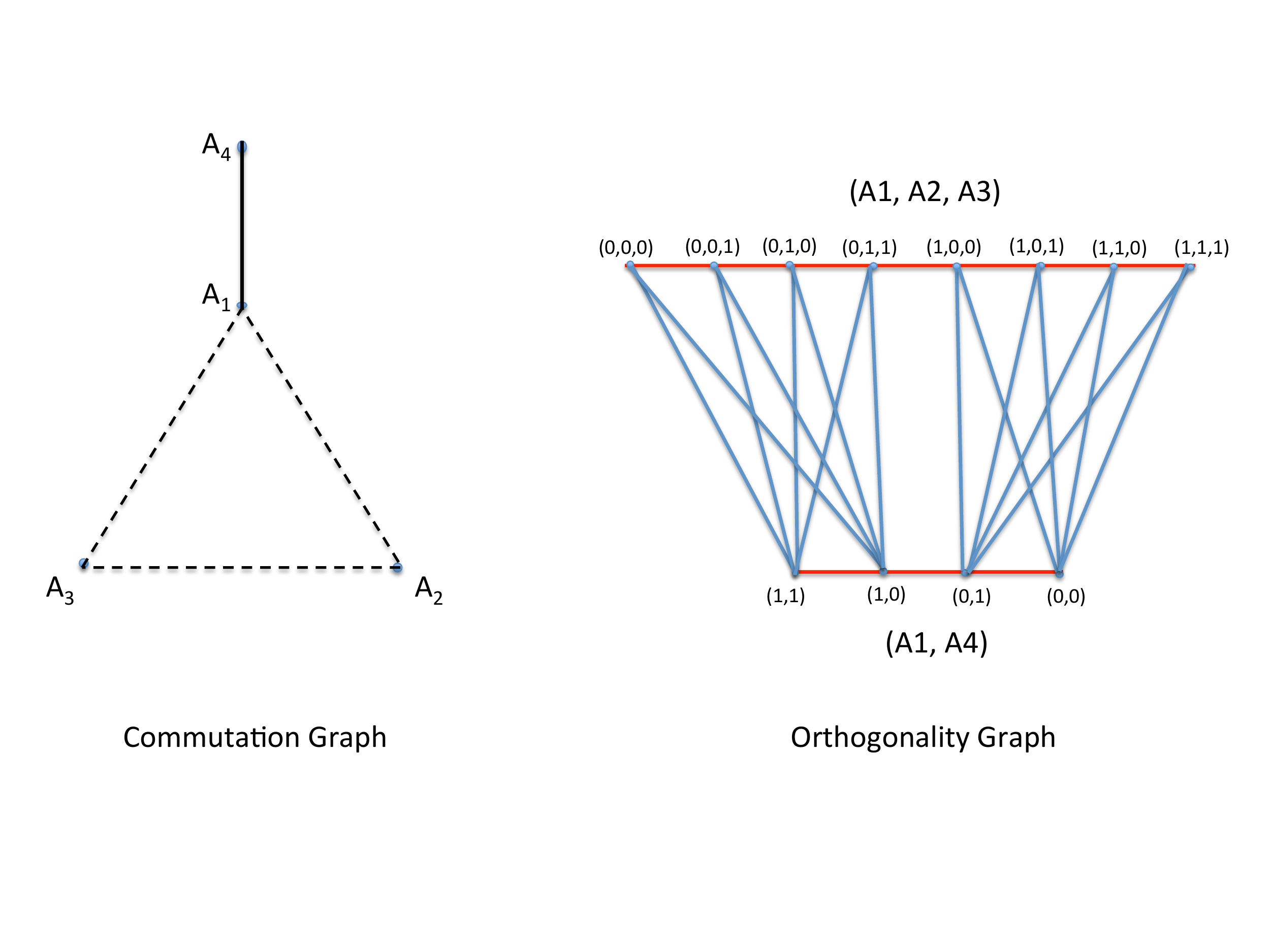}
\caption{(Color online). An example of a single-party XOR scenario with four observables $A_1, A_2, A_3$ and $A_4$. The game imposes a winning constraint of mutual anti-correlations between the observables $A_1, A_2, A_3$ and correlations between $A_1$ and $A_4$. The game is represented here by its commutation graph $G$ (on the left) and the corresponding orthogonality graph $\Gamma$ (on the right) which represents exclusivity relations among events occurring in the game.}
	\label{fig:com-ort}
\end{figure}
\end{center}

It is important to note however that for the general non-contextuality game (both in the XOR and XOR-d scenario) involving general observables in a commutation graph $G$, it is no longer the case that the $\theta_w(\Gamma)$ of the corresponding orthogonality graph gives the quantum value. Instead, we obtain that just as in the case of Bell inequalities, the weighted Lov\'{a}sz theta number only gives an upper bound to the quantum value of general non-contextuality inequalities. Detailed calculations for non-isomorphic graphs for small number of vertices are provided in the next section. In fact as also noted in \cite{Fritz}, in general even for non-contextuality inequalities one needs a hierarchy of semi-definite programs analogous to the well-known semi-definite programming hierarchy \cite{navascues-2008-10} for Bell scenarios, an $n$-partite Bell inequality here being represented by an $n$-partite commutation graph. The analysis of contextuality games via the notion of hyper-graphs with each hyper-edge representing a context was performed in \cite{Fritz} where such an analog of the NPA hierarchy for contextuality was described. 

\section{Linear games for device-independent applications: Pseudo-telepathy}
\label{sec:DIapp}
Linear games are a natural class of Bell inequalities to consider for device-independent applications. Indeed, the class of linear games for binary outcomes (i.e., the XOR games) have been used in most of the device-independent protocols constructed so far, (the CHSH Bell inequality for quantum key distribution \cite{VV12}, the Braunstein-Caves chained Bell inequalities for randomness amplification \cite{CR12} and key distribution against no-signaling adversaries \cite{BHK}, as well as the multi-party XOR games for randomness expansion \cite{MS} as well as randomness amplification \cite{BRGH+13, GM+12}). Linear games have the important property of being \textit{uniform} \cite{KempeRegevToner}, i.e., there exists an optimal quantum strategy for these games where each party's local outcomes are uniformly distributed. This can be seen from the fact that for any quantum strategy for a game with $d$ outcomes, Alice and Bob can make use of a shared random variable $r$ uniformly distributed over $\{0,\dots, d-1\}$ to obtain a quantum strategy with locally random outcomes that achieves the same success probability for the game. Simply Alice performs $a + r \; \text{mod d}$ and Bob performs $b - r \; \text{mod d}$ preserving the value of $a + b  \; \text{mod d}$ while simultaneously randomizing their outcomes. In certain cases, such as the particular example of the CHSH game with ternary outputs in \cite{Liang} or the binary XOR games, locally random (and correlated) outcomes appear naturally in the optimal quantum strategy. As such, it is natural to look for device-independent protocols for randomness or secure key generation that use these Bell inequalities. 

Pseudo-telepathy is an interesting application of quantum correlations to the field of communication complexity. By means of quantum correlations, two (or more) players are able to accomplish a distributed task with no communication at all, which would be impossible using classical strategies alone. Stated in technical terms, these are games $G$ which have $\omega_Q(G) = 1$ but $\omega_C(G) \neq 1$. Pseudo-telepathy games have also found use in certain device-independent protocols \cite{BRGH+13, GM+12} for amplification of arbitrarily weak sources of randomness. In this section, we study the possibility of obtaining pseudo-telepathy within the class of two-party linear games. 

The Braunstein-Caves chained Bell inequalities (which correspond to XOR games for partial functions) have the property that their quantum value approaches $1$ as the number of inputs increases and indeed, this property was very crucial in their use in device-independent applications \cite{CR12, BHK, BKP}. 
%Note that in this case, the classical value also approaches $1$ albeit at a slower rate. 
While one might asymptotically approach unity with increasing number of measurement settings, for real experimental applications, it is extremely important to find Bell inequalities with finite number of inputs and outputs from which randomness or secure key can be extracted. Linear games being the paradigmatic example of Bell inequalities for which optimal quantum strategies involve locally random outcomes, a natural question is to ask whether finite linear games exist which achieve pseudo-telepathy. Our result states that for both total as well as partial functions, while one might asymptotically approach $1$, no finite XOR-d game with prime $d$ number of outcomes exists for which $\omega_q(G) = 1$ while at the same time $\omega_c(G) \neq 1$. This generalizes the recent result for total XOR-d functions in \cite{RAM} and for binary XOR functions in \cite{Cleve}. 
 %Our main result is that for prime values of $d$, no XOR-d games can exhibit perfect pseudo-telepathy.

\begin{theorem}
No finite two-party XOR-d game $G$ corresponding to a (partial or total) function $f(x,y)$ for prime $d$ number of outputs can be a pseudo-telepathy game, i.e., if $\omega_q(G) = 1$, then $\omega_c(G) = 1$. 
\end{theorem}
\begin{proof}
Let $G$ be a finite two-party XOR-d game for prime number of outputs $d$, corresponding to function $f(x,y)$ for input pairs $(x,y)$ and let $\omega_q(G) = 1$. By sharing a uniformly distributed random variable $r$ (specifically by local operations $a + r \; \text{mod d}$ and $b - r \; \text{mod d}$), the two parties Alice and Bob can obtain an optimal quantum strategy which has locally random outputs. Let this optimal quantum strategy be given by $|\psi \rangle \in \mathbb{C}^n \otimes \mathbb{C}^n$ the shared entangled state and $\{\Pi_{x}^{a}\}, \{ \Pi_{y}^{b}\}$ the projectors for inputs $(x,y)$ and outputs $(a,b)$. We have that for this optimal quantum strategy $P_q(a|x) = P_q(b|y) = \frac{1}{d}$ for all $a,x$ and $b,y$. This also implies due to the fact that the XOR-d game is a unique game, that for every input pair $(x,y)$ which has a positive probability in the game, i.e., $\pi(x,y) > 0$, we have
\[ P_q(a,b|x,y) = \begin{cases} 
      \frac{1}{d} & \textrm{ if $a + b$ mod d = $f(x,y)$} \\
      0 & \textrm{ otherwise} \\
   \end{cases} \]
   
Now, as in \cite{RAM} we consider the unitary operators defined as 
$A_x^{k} = \sum_{a=0}^{d-1} \zeta^{-a k} \Pi_{x}^a$ and $B_y^{l} = \sum_{b=0}^{d-1} \zeta^{-b l} \Pi_{y}^{b}$,   
so that we have
\begin{eqnarray}
P_q(a + b \; \text{mod d} = f(x,y) | x,y) = \sum_{k=0}^{d-1} \zeta^{k f(x,y)} \langle A_x^k \otimes B_y^k \rangle. \nonumber \\
\end{eqnarray}
Now, since $\omega_q(G) = 1$ for the game, the above value must equal unity. Putting the above facts together, we have that for every input pair $(x,y)$ with $\pi(x,y) > 0$, there is   
\[ \langle A_x^k \otimes B_y^l \rangle = \begin{cases} 
      \zeta^{-k f(x,y)} & \textrm{ if $k = l$} \\
      0 & \textrm{ otherwise} \\
   \end{cases} \]
Note that for the input pairs $(x,y)$ that do not appear in the game, there is no restriction on the probabilities in the optimal quantum strategy apart from the fact that the local probabilities for Alice and Bob are uniform.  

Now,  following \cite{Cleve} we construct an explicit deterministic (classical) strategy $a:\textsl{X} \rightarrow \{0, \dots, d-1\}$ and $b: \textsl{Y} \rightarrow \{0, \dots, d-1\}$ for Alice and Bob from the above quantum strategy. First, let us fix an orthonormal basis $\{|\phi_1 \rangle, \dots, |\phi_{n^2} \rangle \}$ for $\mathbb{C}^{n} \otimes \mathbb{C}^n$ with $|\phi_1 \rangle = |\psi \rangle$ and the other $|\phi_k\rangle$ chosen to satisfy the orthonormality. Let us define
\begin{eqnarray}
s(x) &:=& \min \{j \in \{2, \dots, n^2\}: \langle \psi | A_x \otimes \mathbf{1} | \phi_j \rangle \neq 0 \}, \nonumber \\ 
t(y) &:=& \min \{j \in \{2, \dots, n^2\}: \langle \psi | \mathbf{1} \otimes B_y^{\dagger} | \phi_j \rangle \neq 0 \}. 
\end{eqnarray}   
With $\lambda$ defined as
\begin{equation}
\lambda(z) =     {d-m+1} \; \text{mod d} \;   \textrm{if $\arg(z)  \in \left[\frac{2(m-1) \pi}{d} , \frac{2m \pi}{d} \right)$, $m \in [d]$ } 
\end{equation}   
we construct the deterministic strategy following \cite{Cleve} as
\begin{eqnarray}
\label{eq:classical-str}
a(x) &:=& \lambda \left(\langle \psi | A_x \otimes \mathbf{1} | \phi_{s(x)} \rangle \right) \nonumber \\
b(y) &:=& d - \lambda \left(\langle \psi | \mathbf{1} \otimes B_y^{\dagger} | \phi_{t(y)} \rangle \right) \; \; \text{mod d} \nonumber \\
\end{eqnarray}
To prove that this classical strategy achieves $\omega_c(G) = 1$ for the game, we have to show that for the quantum strategy these values of $a(x), b(y)$ achieve $P_q(a(x) , b(y) | x, y) = \frac{1}{d}$ when $\pi(x,y) > 0$ so that we have $a(x) + b(y) \; \text{mod d} = f(x,y)$. Evaluating this quantity, we obtain
\begin{equation}
P_q(a(x), b(y) | x,y) = \frac{1}{d^2} \sum_{k=0}^{d-1} \zeta^{k \left(a(x) + b(y) \right)} \langle A_x^k \otimes B_y^k \rangle. 
\end{equation}
Clearly, if $\zeta^{k \left(a(x) + b(y) \right)} \langle A_x^k \otimes B_y^k \rangle = 1$ for all $k$ we achieve $\omega_c(G) = 1$. Suppose by contradiction that $P_q(a(x), b(y) | x,y) = 0$ so that  $\zeta^{\left(a(x) + b(y) \right)} \langle A_x \otimes B_y \rangle = \zeta^t$ for some $t \neq 0$. 
Now, rewriting this by introducing the identity $\mathbf{1} = \sum_{j} | \phi_j \rangle \langle \phi_j |$ we have that
\begin{equation}
\zeta^{\left(a(x) + b(y) \right)} \langle A_x \otimes B_y \rangle = \sum_{j=1}^{n^2} \zeta^{\left(a(x) + b(y) \right)}  \langle \psi | A_x \otimes \mathbf{1} | \phi_j \rangle \langle \phi_j | \mathbf{1} \otimes B_y | \psi \rangle. 
\end{equation} 
Consider the above expression as an inner product of two unit vectors with entries $\zeta^{-a(x)} \langle \phi_j | A_x^{\dagger} \otimes \mathbf{1} | \psi \rangle$ and $\zeta^{b(y)}  \langle \phi_j | \mathbf{1} \otimes B_y | \psi \rangle$. The fact that these are unit vectors follows from $A_x, B_y$ being unitary operators and $\sum_j |\phi_j \rangle \langle \phi_j | = \mathbf{1}$.  We obtain that in order to have $\zeta^{\left(a(x) + b(y) \right)} \langle A_x \otimes B_y \rangle = \zeta^t$, we must have for all $j$
\begin{equation}
\label{eq:cond-noncl}
\zeta^{a(x)} \langle \psi | A_x \otimes \mathbf{1} | \phi_j \rangle = \zeta^{t - b(y)}  \langle \psi | \mathbf{1} \otimes B_y^{\dagger} | \phi_j \rangle.
\end{equation}
Now clearly we have $s(x) = t(y)$ since if $s(x) \neq t(y)$ then when $j$ equals the minimum of these two quantities, one side of the above equation is set to zero while the other is non-zero. But now we observe that for $j=s(x) = t(y)$ and $\arg \left(\langle \psi | A_x \otimes \mathbf{1} | \phi_j \rangle \right) \in [\frac{2(m-1) \pi}{d} , \frac{2m \pi}{d} )$ for some $m \in [d]$ there is
\begin{eqnarray}
&&\arg \left(\zeta^{a(x)} \langle \psi | A_x \otimes \mathbf{1} | \phi_j \rangle \right)  = \nonumber \\ 
&&\frac{2(d-m+1)\pi}{d} +\arg \left(\langle \psi | A_x \otimes \mathbf{1} | \phi_j \rangle \right) \in [0, 2 \pi/d). \nonumber \\ 
\end{eqnarray}
Similarly, for $\arg \left(\langle \psi | \mathbf{1} \otimes B_y^{\dagger} | \phi_j \rangle \right) \in [\frac{2(n-1) \pi}{d} , \frac{2n \pi}{d} )$ for some $n \in [d]$ there is $\arg \left(\zeta^{-b(y)} \langle \psi | \mathbf{1} \otimes B_y^{\dagger} | \phi_j \rangle \right) \in [0, 2 \pi/d)$ so that
Eq.(\ref{eq:cond-noncl}) cannot hold and we have obtained a contradiction. Therefore, we have that $P_q(a(x) , b(y) | x, y) = \frac{1}{d}$ for $\pi(x,y) > 0$ so that the classical strategy given in Eq.(\ref{eq:classical-str}) achieves $\omega_c(G) = 1$. 
\end{proof}

\subsection{Multi-party pseudo-telepathy}
For more than two party non-locality scenarios, the well-known GHZ paradoxes \cite{Mermin} show that it is possible to have XOR games corresponding to partial functions for which $\omega_q(G) = 1$ while $\omega_c(G) < 1$. Indeed, the GHZ paradoxes such as the Mermin inequality have been used in device-independent protocols for randomness amplification \cite{GM+12, BRGH+13} and randomness expansion \cite{MS}. While these involve $m$-party correlation functions, recently it has been of interest to consider Bell inequalities involving two-party correlation functions \cite{TASV+13} that are much easier to measure experimentally. 

As such, we extend the considerations of the previous subsection to the scenario of ``partial" XOR games that involve two-party correlation functions alone and investigate whether pseudo-telepathy is possible in this scenario. These are games for $m$ parties with inputs $(x_1, \dots, x_m)$ and outputs $(a_1, \dots, a_m)$. For each input combination with $\pi(x_1, \dots, x_m) > 0$, there exists a set of pairs $(k, l)$ of parties denoted $S_{(x_1, \dots, x_m)}$ on the XOR of whose outputs the winning constraint depends, i.e., we have that  $V(a_1, \dots, a_m|x_1, \dots, x_m) = 1$ if and only if $a_{k} \oplus a_{l} = f(x_{k}, x_{l})$ for all pairs $(k, l) \in S_{(x_1, \dots, x_m)}$. The Bell inequality thus involves only two-party correlation functions of the type $\langle A^{(k)}_{x_{k}} \otimes A^{(l)}_{x_{l}} \rangle$ where $A^{(i)}_{x_i}$ are observables for party $i$ and input $x_i$ with eigenvalues $\pm 1$. Note that this generalization to many parties is not strictly a unique game since some of parties are not required to output unique outcomes. 

\begin{theorem}
No $m$-party XOR game $G$ involving two-body correlators can be a pseudo-telepathy game, i.e., if $\omega_q(G) = 1$, then $\omega_c(G) = 1$.
\end{theorem}
\begin{proof}
The proof follows similarly to that of the previous theorem. Let $G$ be an $m$-party binary outcome XOR game involving two-body correlation functions and having $\omega_q(G) = 1$. As in the previous theorem, the optimal quantum strategy given by the shared entangled state $|\psi \rangle \in \otimes_{i=1}^{m} \mathbb{C}^n $ and projectors $\{\Pi_{x_{i}}^{a_i}\}$ gives uniform outcomes for each input and each party (obtained for example by each party adding a uniformly distributed $r$ to their outcome), i.e., $P_q(a_i|x_i) = \frac{1}{2}$ for all $a_i,x_i$. While this generalization to many parties is not strictly a unique game so that we cannot precisely identify the non-zero probabilities, we still have for each set of inputs $(x_1, \dots, x_m)$ with $\pi(x_1, \dots, x_m) > 0$ that $\langle A^{(k)}_{x_{k}} \otimes A^{(l)}_{x_{l}} \rangle = (-1)^{f(x_{k}, x_{l})}$ for all pairs of inputs $(k,l) \in S_{(x_1, \dots, x_m)}$. Here, note that $A^{(j)}_{x_j}$ are hermitian operators given by $A^{(j)}_{x_j} = \sum_{a_{j} = 0,1} (-1)^{a_{j}} \Pi_{x_{j}}^{a_{j}}$. This gives that $P_q(a_{k}, a_{l} | x_k, x_l) = \frac{1}{2}$ when $a_{k} \oplus a_{l} = f(x_k, x_l)$ and is $0$ otherwise.

Now, as before following \cite{Cleve} we construct an explicit deterministic (classical) strategy $a_{(i)} :\textsl{X}_i \rightarrow \{0, 1\}$ from the above quantum strategy. We fix an orthonormal basis $\{|\phi_1 \rangle, \dots, |\phi_{n^m} \rangle \}$ for $\otimes_{i=1}^{m} \mathbb{C}^{n}$ with $|\phi_1 \rangle = |\psi \rangle$ and the other $|\phi_k\rangle$ chosen to satisfy the orthonormality. We have
\begin{eqnarray}
s_{(i)}(x_i) &:=& \min \{j \in \{2, \dots, n^m\}: \nonumber \\
&& \; \langle \psi | \mathbf{1}^{\otimes i-1} \otimes A^{(i)}_{x_i} \otimes \mathbf{1}^{\otimes m-i} | \phi_j \rangle \neq 0 \}. 
\end{eqnarray}   
$\lambda$ is now defined as $\lambda(z) = (-1)^k$ if $\arg(z) \in [k \pi , (k+1) \pi)$ for $k = 0, 1$
and the deterministic strategy is given for each party $i \in [m]$ by
%\begin{eqnarray}
%\label{eq:classical-str-2}
$(-1)^{a_{(i)}(x_i)} := \lambda \left(\langle \psi | \mathbf{1}^{\otimes i-1} \otimes A^{(i)}_{x_i} \otimes \mathbf{1}^{\otimes m-i} | \phi_{s_{(i)}(x_i)} \rangle \right)$. 
%\end{eqnarray}
To prove that this classical strategy achieves $\omega_c(G) = 1$, we check that for the quantum strategy these values of $a_{(i)}(x_i)$ achieve $P_q(a_{(k)}(x_{k}), a_{(l)}(x_{l}) | x_k, x_l) = \frac{1}{2}$ for the $(k,l) \in S_{(x_1, \dots, x_m)}$ when $\pi(x_1, \dots, x_m) > 0$. Evaluating this quantity, we get
\begin{eqnarray}
&&P_q(a_{(k)}(x_{k}), a_{(l)}(x_{l}) | x_k, x_l) = \nonumber \\
&&\; \frac{1}{4} \left(1 + (-1)^{\left(a_{(k)}(x_{k}) \oplus a_{(l)}(x_{l}) \right)} \langle A^{(k)}_{x_{k}} \otimes A^{(l)}_{x_{l}} \rangle \right).
\end{eqnarray}
Suppose by contradiction that  $(-1)^{\left(a_{(k)}(x_{k}) \oplus a_{(l)}(x_{l}) \right)} \langle A^{(k)}_{x_{k}} \otimes  A^{(l)}_{x_{l}} \rangle = -1$. 
Now, rewriting this by introducing the identity $\mathbf{1} = \sum_{j} | \phi_j \rangle \langle \phi_j |$ we have that
\begin{eqnarray}
&&(-1)^{\left(a_{(k)}(x_{k}) \oplus a_{(l)}(x_{l}) \right)} \langle A^{(k)}_{x_{k}} \otimes A^{(l)}_{x_{l}} \rangle = \nonumber \\
&&\sum_{j=1}^{n^m} (-1)^{\left(a_{(k)}(x_{k}) \oplus a_{(l)}(x_{l}) \right)} \langle \psi | \mathbf{1}^{\otimes k - 1} \otimes A^{(k)}_{x_{k}} \otimes \mathbf{1}^{m-k} | \phi_j \rangle \nonumber \\
&& \qquad \qquad \langle \phi_j | \mathbf{1}^{\otimes l - 1} \otimes A^{(l)}_{x_{l}} \otimes \mathbf{1}^{\otimes m-l}  | \psi \rangle. 
\end{eqnarray} 
Consider the above expression as an inner product of two unit vectors with entries $(-1)^{a_{(k)}(x_{k})} \langle \phi_j | \mathbf{1}^{\otimes k - 1}  \otimes A^{(k)}_{x_{k}} \otimes \mathbf{1}^{\otimes m-k} |  \psi \rangle$ and $(-1)^{a_{(l)}(x_{l})} \langle \phi_j | \mathbf{1}^{\otimes l - 1} \otimes A^{(l)}_{x_{l}} \otimes \mathbf{1}^{\otimes m-l}  | \psi \rangle$, we must have for all $j$
\begin{eqnarray}
\label{eq:cond-noncl-2}
&&(-1)^{a_{(k)}(x_{k})} \langle \psi | \mathbf{1}^{\otimes k - 1}  \otimes A^{(k)}_{x_{k}} \otimes \mathbf{1}^{\otimes m-k} | \phi_j \rangle = \nonumber \\ &&(-1)^{a_{(l)}(x_{l}) \oplus 1} \langle \psi | \mathbf{1}^{\otimes l - 1} \otimes A^{(l)}_{x_{l}} \otimes \mathbf{1}^{\otimes m-l}  | \phi_j \rangle.
\end{eqnarray}
Now clearly we have $s_{(k)}(x_{k}) = s_{(l)}(x_{l})$ since if $s_{(k)}(x_{k}) \neq s_{(l)}(x_{l})$ then when $j$ equals the minimum of these two quantities, one side of the above equation is set to zero while the other is non-zero. But now we observe that for $j = s_{(k)}(x_{k}) = s_{(l)}(x_{l})$, we have $\arg\left((-1)^{a_{(k)}(x_{k})} \langle \psi | \mathbf{1}^{\otimes k - 1}  \otimes A^{(k)}_{x_{k}} \otimes \mathbf{1}^{\otimes m-k} | \phi_j \rangle \right) \in [0, \pi)$ as well as $\arg \left( (-1)^{a_{(l)}(x_{l})} \langle \psi | \mathbf{1}^{\otimes l - 1} \otimes A^{(l)}_{x_{l}} \otimes \mathbf{1}^{\otimes m-l}  | \phi_j \rangle \right) \in [0, \pi)$ so that Eq.(\ref{eq:cond-noncl-2}) cannot hold and we have obtained a contradiction.
\end{proof}

\section{Explicit examples and Numerical Results}
\label{sec:Comparing-cl-qn}
In this section we provide classical, the almost quantum value (denoted by $\gamma_{3/2}$) \cite{navascues-2008-10} and the quantum values for small XOR-d games in both the single-party contextuality and two-party Bell scenario.  We choose one graph from each equivalence class, since the equivalence relation preserves the classical, quantum and super-quantum values of the game.
% some graphs. We study small GXOR graphs, with 6 or fewer vertices, with 2 and 3 colors, as well as some examples of graphs with uncolored edges.
%We use $\gamma_{3/2}$ as an approximation for quantum. 
%Since $\gamma_{C}\leq\gamma_{Q}\leq\gamma_{3/2}$, it follows that any game with $\gamma_{C}=\gamma_{3/2}$ has $\gamma_{Q}=\gamma_{C}$.  In some cases, however, it happens that $\gamma_{Q}=\gamma_{C}$ even if $\gamma_{3/2}>\gamma_{C}$. Firstly from \cite{RSKK} we know that all chordal graphs, i.e. graphs containing no chordless cycles of length $4$ or more, have $\gamma_{Q}=\gamma_{C}$. Another related upper bound on $\gamma_{Q}$ is the Lov\'{a}sz $\theta$ function of the orthogonality graph representing the game. Clearly, if $\theta=\gamma_{C},$ then $\gamma_{Q}=\gamma_{C}.$ 

%Finally, we need not consider all possible labelings of each graph. Instead, we choose one graph from each equivalence class, since the equivalence relation preserves the classical, quantum and superquantum values of the game.

One interesting sub-class of games is those that have no quantum advantage. Since $\gamma_{C}(G)\leq\gamma_{Q}(G)\leq\gamma_{3/2}(G)$, it follows that any game $G$ with $\gamma_{C}(G)=\gamma_{3/2}(G)$ has $\gamma_{Q}(G)=\gamma_{C}(G)$. Interestingly, we find explicit examples of games where it happens that $\gamma_{Q}(G)=\gamma_{C}(G)$ even though $\gamma_{3/2}(G) >\gamma_{C}(G)$. This makes use of a construction of a joint probability distribution from \cite{RSKK} where it was shown that all chordal graphs, i.e., graphs containing no chordless cycles of length $4$ or more, have $\gamma_{Q}(G)=\gamma_{C}(G)$. Finally, as we have seen in Section \ref{sec:Lovas} the Lov\'{a}sz $\theta$ function of the orthogonality graph representing the game also gives an upper bound on the quantum value that is in general worse than $\gamma_{3/2}(G)$. Clearly, if $\theta(\Gamma(G))=\gamma_{C}(G),$ then $\gamma_{Q}(G)=\gamma_{C}(G).$ 

\subsection{Single-party contextuality XOR games}
\subsubsection{Single color XOR  games}

First, we present the results for the single color XOR games from Section \ref{subsec:single-color-XOR}. We only consider games on connected graphs in which all vertices have degree at least $2$, since for any graph $G=(V,E)$ containing a vertex $v$ of degree $1$ both classical and quantum values are equal to the value for the graph $G'=(V\setminus \{v\},E \setminus \{e\})$ plus $1$ where $e$ is the only edge incident with $v$ in $G$. Since all such graphs with four vertices are classical, we begin with graphs which have five vertices.

\paragraph{Single color XOR games with $5$ vertices}

The only five vertex graphs for which $\gamma_{3/2}\neq \gamma_{C}$ are the cycle $C_5$ ($\gamma_{C}=4, \gamma_{3/2}\approx 4.472$) 
%4.47213595$) 
and the complete graph $K_5$ ($\gamma_{C}=6, \gamma_{3/2} \approx 6.25$).
% see Fig. \ref{fig:5kol}  
Since the classical and quantum values of any complete graph must be equal, this means $C_5$ is the only graph for which $\gamma_{Q}\neq\gamma_{C}.$

%\begin{center}
%\begin{figure}[h!]
%		\includegraphics[width=0.4\textwidth]{1kol5.png}
%\caption{The only non-equivalent single-color graphs with 5 vertices.}
%	\label{fig:5kol}
%\end{figure}
%\end{center}

\paragraph{Single color XOR games with $6$ vertices}

Out of the 61 non-isomorphic graphs with six vertices, four have $\gamma_{3/2}>\gamma_{C}$, see Fig. \ref{fig:1kol6}.

\begin{center}
\begin{figure}[h!]
		\includegraphics[width=0.4\textwidth]{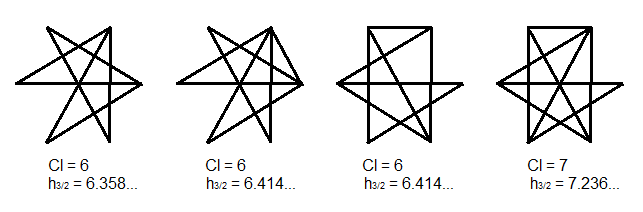}
\caption{The only non-equivalent single-color graphs with 6 vertices for which the classical value is not equal to $\gamma_{3/2}$. Note that all the edges here denote anti-correlations, for simplicity, the dashed edges have been replaced by solid ones.}
	\label{fig:1kol6}
\end{figure}
\end{center}

\paragraph{Single color XOR games with $7$ vertices}
Out of 507 analyzed graphs, 54 have $\gamma_{3/2}>\gamma_{C}$. For four out of these $\gamma_{3/2}=\gamma_{C}+0.25.$ All edges in those graphs lie in cliques of size $3$ or more, and we construct an explicit joint probability distribution following \cite{RSKK} which implies that $\gamma_{Q}=\gamma_{C}$. 
%Each of those graphs also contains $K_5$ as an induced subgraph. 
%We pass now to study $\beta_Q(G,K)$. 
%[In the Appendix, we show detailed result for the method of hierarchies \cite{navascues-2008-10} for classes of graphs.] 
%We list some of them below. \kh{
It is important to note, that in Fig. \ref{fig:1kol6} below we use colors for different purpose than in other figures, that is not to depict one of the 3 kinds of permutations as in Fig. \ref{fig:3colors}, but to visualize certain subgraphs of a given graph.
%}

\begin{center}
\begin{figure}[h!]
		\includegraphics[width=0.4\textwidth]{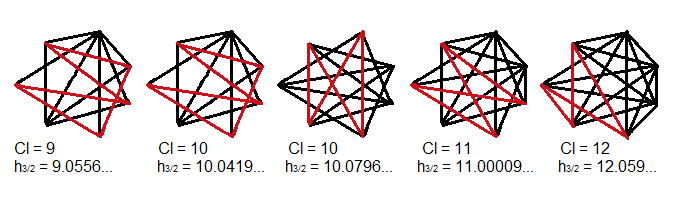}
		\includegraphics[width=0.4\textwidth]{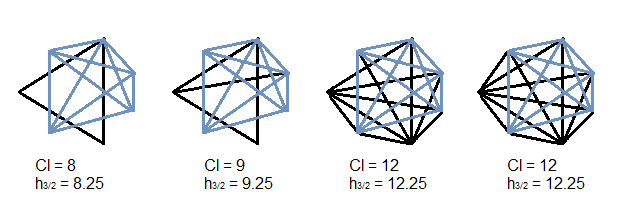}
\caption{Colors on these figures are used only in order to visualize certain subgraphs. Figure
(a) depicts non-equivalent single color graphs, with very small difference between $\gamma_{3/2}$ and $\gamma_{C}$ compared to the 
typical case. All these graphs contain a chordless cycle of length at least 4, marked in red. Figure (b) depicts all the non-equivalent single-color graphs with 7 vertices for which $\gamma_{3/2}=\gamma_{C}+0.25$. Note that all these graphs are chordal, and admit a joint probability distribution so that the $\gamma_{Q} = \gamma_{C}$ for these graphs. They also all contain $K_5$ as an induced subgraph, marked in blue.}
	\label{fig:1kol6}
\end{figure}
\end{center}

\subsubsection{Two and three color XOR-3 games}
\label{subsec:2and3colors}
We have calculated the classical and almost quantum values for all equivalence classes of 3 color (XOR-3) games defined by small connected graphs without vertices of degree 1. Adding such a vertex to any graph $G$ simply increases both classical and quantum values by 1, since the additional constraint is always satisfiable. Every XOR-3 game graph with five or less vertices for which the values are different is equivalent to one of the graphs in Fig \ref{fig:3koldo5w}.

\begin{figure}[h!]
	\includegraphics[width=0.4\textwidth]{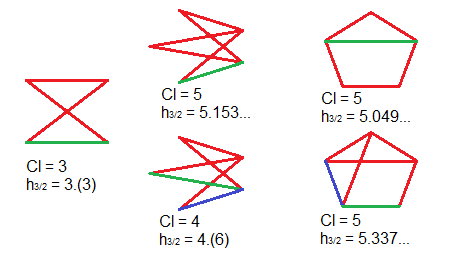}
	\caption{All non-equivalent graphs on 5 or less vertices with no vertices of degree 1 in which $\gamma_{3/2}\neq\gamma_{Cl}$. Bipartite graphs (the first three) represent Bell's inequalities, others correspond to contextual games.}
	\label{fig:3koldo5w}
\end{figure}

\subsection{Two-party XOR-3 Bell inequalities}
\subsubsection{Total function ternary input XOR-3 games}
Every bipartite 3 color (XOR-3) game on six vertices for which the $\gamma_{3/2}$ value is higher than classical is equivalent to one of the graphs in Fig \ref{fig:3kol6wB}. In this case, we have also calculated the quantum value by optimizing over two-qutrit states $\sum_{i, j = 0}^{2} \alpha_{i,j} | i, j \rangle$ and observables. In each case of ternary input-output Bell inequalities, except the CHSH-3 scenario considered in \cite{BavarianShor, Liang} we find that the quantum value calculated for qutrits matches (up to numerical precision) the almost quantum value of the SDP hierarchy \cite{navascues-2008-10}. 

\begin{figure}[h]
	\includegraphics[width=0.40\textwidth]{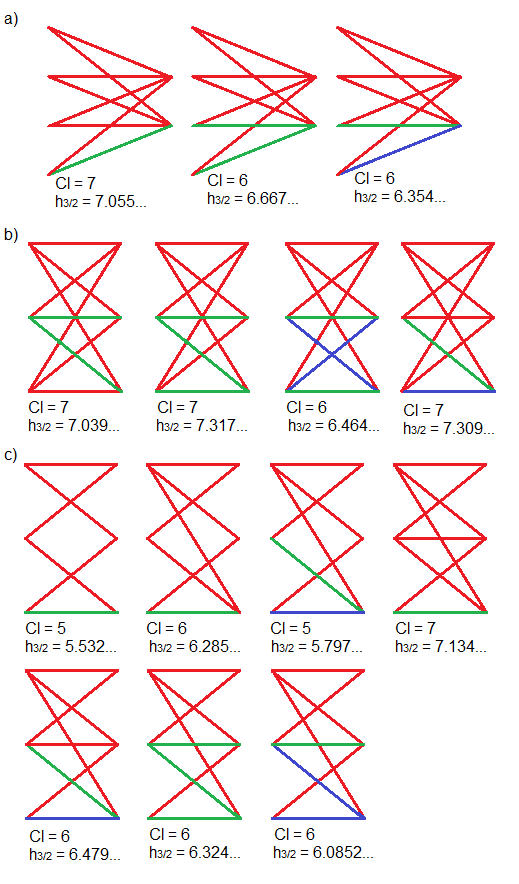}
	\caption{All non-equivalent bipartite graphs with 6 vertices (and no vertices of degree 1) with $\gamma_{3/2}\neq \gamma_{C}$. In each of these cases, an optimization over two-qutrit states and observables shows \\a) $K_{4,2}$ Bell's inequalities \\b) $K_{3,3}$ Bell's inequalities. Note that the third graph from left falls into the CHSH-d class of Bell inequalities considered in \cite{BavarianShor, Liang}. \\c) Other graphs.\\\textit{In each of the games in (b), (c) except the CHSH-3 game, an optimization over two-qutrit states and observables shows that the quantum value is in fact equal to the $\gamma_{3/2}$ value up to numerical precision}.}
	\label{fig:3kol6wB}
\end{figure}

\subsubsection{Partial function ternary input XOR-3 games}

We have also calculated classical and $\gamma_{3/2}$ values for some small (5 vertices and bipartite with 6 vertices) XOR-3 game graphs with uncolored edges, i.e., those corresponding to partial functions. 
%A well-known subclass of this type of game are complete bipartite XOR graphs with some uncolored edges - the classic Bell scenario with binary outcomes for which classical and quantum values are known \cite{Cleve}. Here we consider XOR-3 graphs for partial functions with 3 colors for which $\gamma_{3/2}>\gamma_{C}$. 
We conjecture that these are the only 3-colored graphs with uncolored edges for which classical and quantum values may be different. Fig. \ref{fig:grayL3} depicts all possibly nonclassical classes of 3-colored graphs with 5 vertices, and bipartite graphs with 6 vertices, in which every vertex is incident to at least two colored edges. 
%For comparison, \ref{fig:grayL2} depicts two-colored (XOR) graphs.

%\begin{figure}[h]
%	\includegraphics[width=0.40\textwidth]{grayL2.png}
%	\caption{XOR graphs with uncolored edges such that $\gamma_{3/2}\neq\gamma_{C}$. Note that the set includes only one chain.}
%	\label{fig:grayL2}
%\end{figure}

\begin{figure}[h]
	\includegraphics[width=0.40\textwidth]{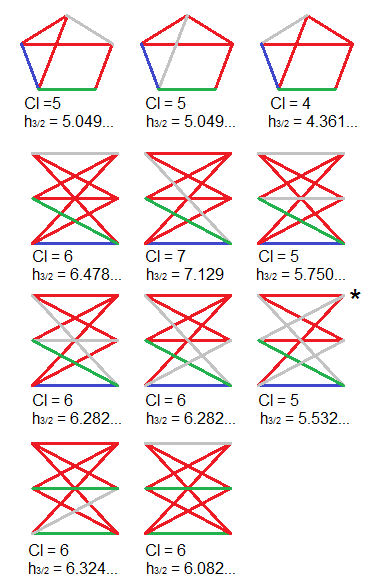}
	\caption{Three-colored XOR-d graphs with uncolored edges such that $\gamma_{3/2}\neq\gamma_{C}$. Note, like in the two-colored case, that the set includes only one chain (marked with "*").}
	\label{fig:grayL3}
\end{figure}

Note that the set only includes one chain (i.e. the graph $K_{3,3}$ in which only a 6-cycle is colored). 
%All labelings of $C_6$ with two colors are equivalent to either $K(e)=\id$ for all $e\in E,$ which is good, or $K(e_1)=(01)$ and $K(e)=\id$ for $e\in E-\{e_1\}$ (since $KC_6$ is isomorphic to either $C_{12}$ or two disconnected 6-cycles). Similarly, 
All labelings of $C_6$ with three colors are equivalent to either $K(e)=(01)$ for all $e\in E$ (good) or $K(e_1)=(12)$ and $K(e)=(01)$ for $e\in E-\{e_1\}.$ (Interestingly, this is not necessarily the case for 4 and more colors) Thus, all labelings of the same graph with 2 or three colors in which only a 6-cycle is colored must also form exactly two equivalence classes. As explained in the beginning of this section, vertices of degree 1 do not matter in graphs for total function games. However, even though we do not count uncolored edges as constraints, vertices incident to one or more uncolored edge and only one colored edge do need to be considered. If $v\in V(G)$ is incident to only one colored edge, the classical value of the game is equal to $\gamma_{C}(G-\{v\},K)+1$ The quantum value, however, may differ from $\gamma_{Q}(G-\{v\})+1.$ An example is presented in Fig. \ref{fig:legs}.

\section{Conclusions}
We have studied the generalization of XOR games to arbitrary number of outcomes known as linear games for prime or prime power outputs or more generally called XOR-d games. We first abstracted two paradigmatic properties of the XOR games and showed that for odd values of $d$, the unique class of games that obey these two properties were the earlier studied class of linear games. In both the contextuality and non-locality scenarios, we introduced a graph-theoretical description of these games in terms of edge labelings with colors representing different permutations. There followed a natural relation between equivalent classes of games and the graph-theoretic notion of switching equivalence and signed graphs. We also studied the classical value of these games in terms of graph-theoretic parameters. In particular, computing the classical value of single-party anti-correlation XOR games was related to finding the edge bipartization number of a graph, which is known to be MaxSNP hard. Computing the classical value of more general XOR-d games was related to the identification of specific bad and ugly cycles in the graph. Studying classical value can be done in many ways, in particular here we have studied it via {\it three} types of cycles in a graph - the so called good cycles which satisfy all vertex assignments, the bad cycles for which no assignment leads to satisfiability, and interestingly {\it the ugly} ones, which makes the problem of satisfaction difficult, as they satisfy some but not all vertex assignments. Another graph theoretical tool is the graph $KG$ - a permutation graph of the game graph $G$. This tool will be heavily used in \cite{RS}, here we showed that it allows for testing whether a given graph corresponds to a game that can be won with probability $1$ using classical resources. 
% probability $1$ of winning it with help of classical resources.

We also studied the quantum value of these games using the Lov\'{a}sz theta number of the corresponding orthogonality graph. We show how the constraint graph representing the game can be used to construct the orthogonality graph and find that its Lov\'{a}sz theta number still gives only an upper bound on the quantum value even for single-party contextuality XOR-d games. An important property of the XOR-d game Bell inequalities is that for these, an optimal quantum strategy can be found for which the outcomes of each party are uniformly distributed and correlated. This makes these games ideal candidates for device-independent applications. Indeed XOR games, in particular the Braunstein-Caves chained Bell inequalities have found widespread use in such tasks. We showed that for both partial and total functions, no finite XOR-d game (for prime number of outcomes) exhibits the property of pseudo-telepathy, i.e., maximum algebraic violation of such Bell inequalities cannot be obtained in quantum theory. We also extended the result to multi-party "partial" XOR games which involve only two-body correlation functions, showing that such Bell inequalities cannot achieve algebraic violation. 

An interesting question is to develop this framework to get more analytical bounds such as in \cite{BavarianShor}. It would also be important to study more general unique games using a similar approach. Given that finite XOR-d games do not exhibit pseudo-telepathy, an important open question is whether the chained Bell inequalities and their generalization to many outcomes are the class of XOR-d games that exhibit the best asymptotic rate of convergence of the quantum value to unity. Numerical studies for small size games indicates that apart from the CHSH-d games considered earlier, the quantum value for ternary output games is achieved at the level $1+AB$ of the SDP hierarchy from \cite{navascues-2008-10}. It would be interesting to investigate whether a sub-class of the XOR-d games can be proved to achieve optimality at particular intermediate levels of the hierarchy. 
%
%
%It would be also important to study other non-GXOR, but still unique classical and quantum games using similar approach.
%Let us recall, that in our approach we based on two properties of $XOR$ games which are (1) permutations (which correspond to constraints) are equal to their inverse (2) each pair $(a,\pi(a))$ appears exactly once in the set of all permutations (constraints). However we have considered only permutations defined as $\pi(x) = x +i$ for $i \in [d], x \in [d]$. It would 
%be interesting to find if this set of permutations is the only one, which satisfy these two properties of the XOR game. We leave as an open problem.
%Last but not least our generalization of XOR to GXOR games is not unique, studying other approaches may be fruitful as well.}
%\com{do we say about the problem of general commuting versus tensor product commuting ?}

%\appendix{Appendix}

%********************
%To be more precise, we now formally define the GXOR game. This game is represented by a graph $G=(E,V)$. The values assigned to vertices belong to the set $[d]=\{0,1,...,d-1\}$ and the edges are labeled with colors, i.e. permutations belonging to the set $L_d$. A vertex-assignment $k:V(G)\mapsto\{0,1,...,d-1\}$ is \textit{consistent} if it has no contradictions (i.e. if the permutation assigned to the edge $(u,v)$ is $\pi$, then $\pi(u)=v$, for all $u,v\in V$) and the Bell's/contextual inequality has the form $p_{win}\leq 1-\frac{\beta_C(G)}{\left|E(G)\right|}.$

%\mh{commuting versus tensor product}

{\em Acknowledgements.} 
%We acknowledge useful discussions with Wojciech Wantka. This work is supported by the ERC AdG grant QOLAPS, EC grant RAQUEL and also forms part of the Foundation for Polish Science TEAM project co-financed by the EU European Regional Development Fund. 
We acknowledge useful discussions with Ryszard Horodecki and Wojciech Wantka. This work is supported by the EC IP QESSENCE, ERC AdG QOLAPS, EU grant RAQUEL and the Foundation for Polish Science TEAM project co-financed by the EU European Regional Development Fund. Simone Severini is supported by the Royal Society and EPSRC.

\end{document}